\documentclass[prl, twocolumn, aps, preprintnumbers,showpacs]{revtex4}

\usepackage{amsfonts,amssymb,amsmath,amsthm,graphicx}

\usepackage[usenames,dvipsnames]{color}
\usepackage{epic} 
\usepackage{pict2e}   
\usepackage{calc}
\usepackage{times}

\newtheorem{texttheorem}{Theorem}
\newtheorem{textcorollary}{Corollary}
\newtheorem{theorem}{Theorem}
\newtheorem{lemma}{Lemma}

\newtheorem{corollary}{Corollary}

\newcommand*{\bbN}{\mathbb{N}}
\newcommand*{\bbC}{\mathbb{C}}
\newcommand*{\bbR}{\mathbb{R}}
\newcommand*{\bbZ}{\mathbb{Z}}

\newcommand*{\cH}{\mathcal{H}}

\newcommand*{\cK}{\mathcal{K}}

\newcommand*{\cP}{\mathcal{P}}

\newcommand*{\cS}{\mathcal{S}}
\newcommand*{\cU}{\mathcal{U}}
\newcommand*{\cT}{\mathcal{T}}

\newcommand*{\id}{\mathbf{1}}
\newcommand*{\complex}{\mathbb{C}}

\newcommand*{\la}{\lambda}

\newcommand*{\tr}{\mathrm{tr}}
\newcommand*{\ket}[1]{| #1 \rangle}
\newcommand*{\bra}[1]{\langle #1 |}
\newcommand*{\spr}[2]{\langle #1 | #2 \rangle}
\newcommand*{\braket}[2]{\langle #1 | #2 \rangle}
\newcommand*{\proj}[1]{\ket{#1}\!\bra{#1}}

\newcommand*{\diag}{\text{diag}}

\newcommand*{\spanv}{\mathrm{span}}

\newcommand*{\eps}{\varepsilon}
\newcommand*{\Sym}{\mathrm{Sym}}

\newcommand*{\fail}{\mathrm{fail}}

\newcommand*{\naturals}{\bbN}

\newcommand*{\CP}[1]{\mathbb{C}P^{#1}}

\newcommand*{\one}{\mathbf{1}}

\newcommand*{\poly}{\mathrm{poly}}

\newcommand*{\mult}{\mathrm{mult}}

\newcommand*{\CG}[6]{\left[\begin{array}{ccc} #1 & \ #2 &  #3 \\ #4 &
      #5 &  #6 \end{array}\right]}

\newcommand*{\CGG}[6]{\left\{\begin{array}{ccc} #1 & \ #2 &  #3 \\ #4 &  #5 &  #6 \end{array}\right\}}

\begin{document}

\title{Reliable Quantum State Tomography}  

\author{Matthias \surname{Christandl}}  \author{Renato \surname{Renner}}
\affiliation{Institute for Theoretical Physics, ETH Zurich,
  Wolfgang-Pauli-Strasse 27, CH-8093 Zurich, Switzerland}

\pacs{03.65.Wj, 02.50.-r, 03.67.-a}

\date{November 7, 2012}
\begin{abstract}

  Quantum state tomography is the task of inferring the state of a
  quantum system by appropriate measurements.  Since the frequency
  distributions of the outcomes of any finite number of measurements
  will generally deviate from their asymptotic limits, the estimates
  computed by standard methods do not in general coincide with the
  true state, and therefore have no operational significance unless
  their accuracy is defined in terms of error bounds.  Here we show
  that quantum state tomography, together with an appropriate data
  analysis procedure, yields reliable and tight error bounds,
  specified in terms of confidence regions|a concept originating from
  classical statistics. Confidence regions are subsets of the state
  space in which the true state lies with high probability,
  independently of any prior assumption on the distribution of the
  possible states. Our method for computing confidence regions can be
  applied to arbitrary measurements including fully coherent ones; it
  is practical and particularly well suited for tomography on systems
  consisting of a small number of qubits, which are currently in the
  focus of interest in experimental quantum information science.
  
  \end{abstract}

\maketitle

The state of a classical system can in principle be determined to
arbitrary precision by applying a single measurement to it.  Any
imprecisions are due solely to inaccuracies of the measurement
technique, but not of fundamental nature.  This is different in
quantum theory. It follows from Heisenberg's uncertainty
principle 
that measurements generally have a random component and that
individual measurement outcomes only give limited information about
the state of the system|even if an ideal measurement device is
used. To illustrate this difference, it is useful to take an
information-theoretic perspective. Assume, for instance, that we are
presented with a two-level system about which we have no prior
information except that it has been prepared in a pure state, and our
task is to determine this state. If the system was classical, there
are only two possible pure states, and one single bit of information
is therefore sufficient for its full description. Furthermore, a
single measurement of the system suffices to retrieve this bit. If the
system was quantum, however, the situation becomes more interesting. A
two-level quantum system (a qubit) admits a continuum of pure states
that can, for example, be parameterized by a point on the Bloch
sphere. To determine this point to a given accuracy $\Delta$, at least
$\log_2 (4/\Delta^2)$ bits of information are necessary~\footnote{A
  disc with (great-circle) radius $\Delta$ on the Bloch sphere has
  area $2 \pi (1- \cos \Delta) \leq \pi\Delta^2 $, whereas the full
  Bloch sphere has area $4 \pi$. Consequently, there are at least $(4
  \pi) /(\pi\Delta^2) = 4/\Delta^2$ such discs. Note also that the
  (great-circle) distance $\Delta$ between two pure states $\phi$ and
  $\psi$ is related to their fidelity, $F(\phi, \psi) =
  |\braket{\phi}{\psi}| = |\cos \frac{\Delta}{2}|$, as well as to
  their trace distance, $\| \proj{\psi} - \proj{\phi}\|_1 = 2 |\sin
  \frac{\Delta}{2}| \approx \Delta$.}.  Conversely, according to
Holevo's bound~\cite{Holevo73}, any measurement applied to a single
qubit will provide us with at most one bit of information. And even if
$n$ identically prepared copies of the qubit were measured, at most
$\log_2 (n+1)$ bits of information about their state can be
obtained~\footnote{The bound follows from the fact that the joint
  state of $n$ identically prepared copies of a pure state in $\cH =
  \mathbb{C}^2$ lies in the symmetric subspace of $\cH^{\otimes n}$,
  which has dimension $n+1$.}. Hence the accuracy, $\Delta$, to which
the state can be determined always remains finite ($\Delta \geq
\frac{2}{\sqrt{n+1}}$), necessitating the specification of error bars.

The impact that randomness in measurement data has on the accuracy of
estimates has been studied extensively in statistics and, in
particular, estimation theory~\cite{Kay93}. 
The latter is concerned with the general problem of estimating the
values of parameters from data that depend probabilistically on
them. The data may be obtained from measurements on a quantum system
with parameter-dependent state, as considered in quantum estimation
theory~\cite{Helstrom76}. Quantum state tomography can be seen as a
special instance of quantum estimation, where one aims to estimate a
set of parameters large enough to determine the system's state
completely~\cite{Fano57,VogRis89,Raymeretal93,LePaAr95,Hradil97,Arianoetal03,Sugiyama11}.

 An obvious choice of parameters are the matrix elements of a density
 operator representation of the state. Due to the finite accuracy,
 however, the individual estimates for the matrix elements do not
 generally correspond to a valid density operator (for instance, the
 matrix may have negative eigenvalues).  This problem is avoided with
 other techniques, such as maximum likelihood estimation
 (MLE)~\cite{Hradil97,Banaszeketal99,Hradiletal04}, which has been
 widely used in
 experiments~\cite{Kwiatetal01,Blattetal04,Zeilingeretal05,BlattWineland08,Wallraffetal09,Home09,
   ZollerBlattetal11}, or Bayesian
 estimation~\cite{Helstrom76,Jones91,Buzeketal98,Schacketal01,TanKom05,Blume10,
   kalman}.

 In MLE, an estimate for the error bars can be obtained from the width
 of the likelihood function, which is approximated by the Fisher
 information matrix \cite{Banaszeketal99, usami, hradilmogil, burgh,
   hradil-diag, Sugiyama11}. In current experiments one also uses
 numerical plausibility tests known as ``bootstrapping'' or, more
 generally, ``resampling''~\cite{Efron93, Home09} in order to obtain
 bounds on the errors. However, despite being reasonable in many
 practical situations, these bounds are not known to have a
 well-defined operational interpretation and, in the case of the
 resampling method, may lead to an underestimate of the
 errors~\cite{Blume12}.

 In contrast, Bayesian methods can be used to calculate ``credibility
 regions'', i.e., subsets of the state space in which the state is
 found with high probability. This probability, however, depends on
 the choice of a ``prior'', corresponding to an assumption about the
 distribution of the states before the measurements (in particular,
 the assumption can not be justified by the experimental
 data). Furthermore, we remark that most known techniques are based on
 the assumption of independent and identical measurements (a notable
 exception is the one-qubit adaptive tomography analysis
 of~\cite{Sugiyama12}).  We refer to~\cite{Blume10} for a further
 discussion of currently used approaches to quantum state tomography,
 including pedagogical examples illustrating their limitations.

 In this Letter we introduce a method to obtain \emph{confidence
   regions}, that is, regions in state space which contain the true
 state with high probability. A point in the region may then serve as
 estimate and the maximal distance of the point to the border of the
 region as error bar. Our method allows to analyse data obtained from
 arbitrary quantum measurements, including fully coherent ones. The
 method does not rely on any assumptions about the prior distribution
 of the states to be measured. This makes it highly robust so that it
 can, for instance, be applied in the context of quantum cryptography,
 where the states to be estimated are chosen adversarially.

 The remainder of this Letter is organized as follows. We first
 describe a very general setup for tomography of quantum states
 prepared in a sequence of experiments, where we do not make the
 typical assumption that the states are independent and identically
 distributed (i.i.d.). We then show that, nevertheless, properties of
 the states can be inferred reliably using a suitable tomographic data
 analysis procedure (Theorem~\ref{thm:main}). In motivation and
 spirit, this result relates to recent research efforts on quantum de
 Finetti representations. We then specialise our setup to the case
 where, in principle, the experiments may be run an arbitrary number
 of times (while still only finitely many runs are used to generate
 data). This special case is (by the quantum de Finetti theorem)
 equivalent to an i.i.d.~preparation of the states, thereby justifying
 the common i.i.d.\ assumption in data analysis. The theorem, applied
 to this special case, then results in a construction for confidence
 regions for quantum state tomography (Corollary~\ref{cor:main}).

 \emph{General Scenario.|}Consider a collection $\cS_1, \ldots,
 \cS_{n+k}$ of finite-dimensional quantum systems with associated
 Hilbert space $\cH$, as depicted in Fig.~\ref{fig_setup} (see
 also~\cite{Renner07} and~\cite{chiri-defin}, where a similar setup is
 considered). We denote by $d$ the dimension of $\mathcal{H}$. For
 example, one may think of $n+k$ particles prepared in a series of
 experiments, where $\cH$ could correspond to the spin degree of
 freedom. From this collection, a \emph{sample} consisting of $n$
 systems is selected at random and measured according to an
 (arbitrary) Positive Operator Valued Measure (POVM) $\{B^n\}$, a
 family of positive semi-definite operators $B^n$ on $\cH^{\otimes n}$
 such that $\sum_{B^n} B^n = \id_{\cH}^{\otimes n}$. That is, each
 POVM element $B^n$ corresponds to a possible sequence of outcomes
 resulting from (not necessarily independent) measurements on the $n$
 systems.  The goal of quantum state tomography is to infer the state
 of the remaining $k$ systems, using the outcomes of these
 measurements.

 Note that the $k$ extra systems are not measured during data
 acquisition. Nevertheless, they play a role in the above scenario, as
 they are used to define operationally what state we are
 inferring. (In the special case of i.i.d.\ states, the extra systems
 are simply copies of the measured systems|see below.)  We also remark
 that, instead of measuring a sample of $n$ systems chosen at random,
 one may equivalently permute the initial collection of $n+k$ systems
 at random and then measure the first $n$ of them, i.e., $\cS_1,
 \ldots, \cS_{n}$. We will use this alternative description for our
 theoretical analysis.

\begin{figure}
\includegraphics[width=0.95\columnwidth,clip,trim = 25mm 206mm 111mm 18mm]{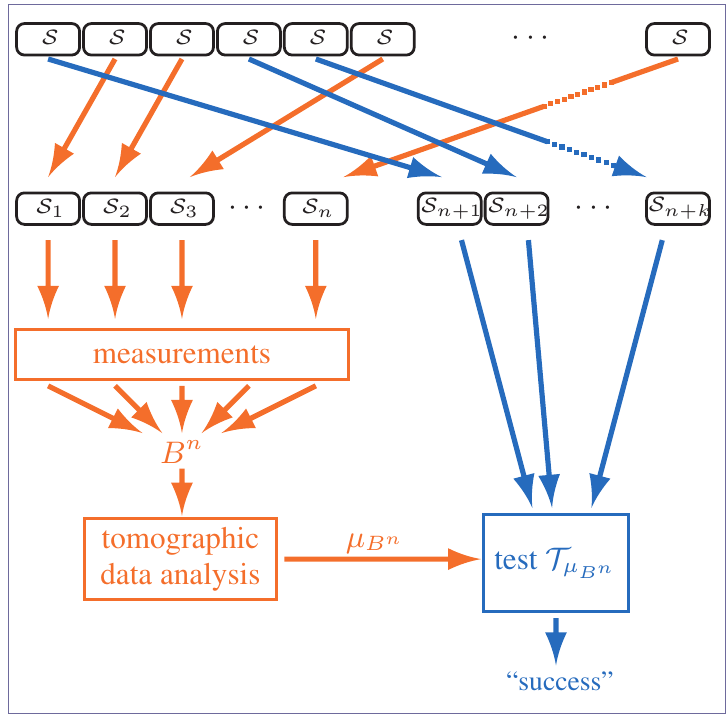} 
\caption{{\bf General scenario.}  Measurements are applied to a sample
  $\cS_1, \ldots, \cS_n$ consisting of $n$ systems, chosen at random
  from a collection of $n+k$ systems. The outcomes of the measurements
  are collected and given as input, $B^n$, to a data analysis
  procedure (orange). The aim of quantum state tomography is to make
  reliable predictions about the state of the remaining $k$
  (non-measured) systems $\cS_{n+1}, \ldots, \cS_{n+k}$. To model such
  predictions, we consider hypothetical tests, which output
  ``success'' whenever their input has a desired property
  (blue). Given only the output of the data analysis procedure,
  $\mu_{B^n}$, it is possible to characterize the tests
  $\cT_{\mu_{B^n}}$ that are passed with high
  probability|independently of the initial state of the $n+k$ systems
  (Theorem~\ref{thm:main}). \label{fig_setup} }
\end{figure}

In order to describe our main results, we imagine that the measurement
outcomes $B^n$ are processed by a data analysis routine that outputs a
probability distribution $\mu_{B^n}$ on the set of mixed states,
defined by
\begin{align*}
  \mu_{B^n}(\sigma)d\sigma = \frac{1}{c_{B^n}} \tr[\sigma^{\otimes n}
  B^n] d\sigma 
\end{align*}
(see Fig.~\ref{fig_estimateillustration} for an illustration).  Here
$d\sigma$ denotes the Hilbert-Schmidt measure with $\int
d\sigma=1$. Furthermore, $c_{B^n} =\tr[B^n \otimes \id_{\cK}^{\otimes
  n} \cdot \id_{\Sym^n(\cH \otimes \cK)}] / \binom{n+d^2-1}{d^2-1}$ is
a normalisation constant, where $\cK\cong \cH\cong \complex^d$ and
where $\id_{\Sym^n(\cH \otimes \cK)}$ is the projector onto the
symmetric subspace of $(\cH \otimes \cK)^{\otimes n}$.  Note that, in
Bayesian statistics, $\mu_{B^n}(\sigma)d\sigma$ corresponds to the
\emph{a posteriori} distribution when updating a Hilbert-Schmidt prior
$d \sigma$. Furthermore, in MLE, $\sigma \mapsto \tr[\sigma^{\otimes
  n} B^n]$ is known as the \emph{likelihood function}. Since our work
is not based on either of these approaches, however, we will not use
this terminology and simply refer to $\mu_{B^n}$.

\bigskip \emph{Reliable Predictions.|}We now show that $\mu_{B^n}$
contains all information that is necessary in order to make reliable
predictions about the state of the remaining systems $\cS_{n+1},
\ldots, \cS_{n+k}$. To specify these predictions, we consider
\emph{hypothetical tests}, a quantum version of a similar concept used
in classical statistics. Any such test acts on the joint system
consisting of $\cS_{n+1}, \ldots, \cS_{n+k}$ (see
Fig.~\ref{fig_setup}). Mathematically, a test is simply a measurement
with binary outcome, ``success'' or ``failure'', specified by a joint
POVM $\{T^{\fail},\id_{\cH}^{\otimes k} -T^{\fail}\}$ on $\cH^{\otimes
  k}$~\footnote{For our technical treatment (see \emph{Supplemental
    Information}), we also consider tests that act on a larger space,
  $({\cH \otimes \cK})^{\otimes k}$, which includes purifications of
  the systems.}. Note that the state of $\cS_{n+1}, \ldots, \cS_{n+k}$
could be inferred if we knew which hypothetical tests it would pass.
Hence, instead of estimating this state directly, we can equivalently
consider the task of predicting the outcomes of the hypothetical
tests.

Assume now that we carry out a test $\cT_{\mu_{B^n}} =
\{T^{\fail}_{\mu_{B^n}},{\id_{\cH}^{\otimes k}
  -T^{\fail}_{\mu_{B^n}}}\}$ depending on $\mu_{B^n}$. We denote by
$\rho^{n+k}$ the (unknown) joint state of the systems $\cS_1, \ldots,
\cS_{n+k}$ before the tomographic measurements. (As described above,
we can assume without loss of generality that the systems are permuted
at random, so that $\rho^{n+k}$ is permutation invariant.) If the
outcome of the tomographic measurement is $B^n$, then the
post-measurement state of the remaining systems is given explicitly by
$\rho_{B^n}^k = \frac{1}{\tr[B^n \rho^{n}]} \tr_{\cH^{\otimes n}}[B^n
\otimes \id_{\cH }^{\otimes k} \cdot\rho^{n+k}], $ where
$\tr_{\cH^{\otimes n}}$ denotes the partial trace over the $n$
measured systems.  Hence, the probability that the test
$\cT_{\mu_{B^n}}$ fails for the above state $\rho_{B^n}^k $ equals
$\tr[T^{\fail}_{\mu_{B^n}} \rho_{B^n}^k]$. The following theorem now
provides a criterion under which this failure probability is upper
bounded by any given $\eps > 0$. Crucially, the criterion only depends
on $\mu_{B^n}$, which is obtained by the tomographic data analysis. In
other words, $\mu_{B^n}$ allows us to determine which hypothetical
tests the state $\rho_{B^n}^k$ would pass.

\begin{figure}
\begin{center}
\setlength{\unitlength}{0.11mm} 
\centering
\includegraphics[width=0.44\columnwidth,clip,trim=30mm 220mm 130mm
 25mm]{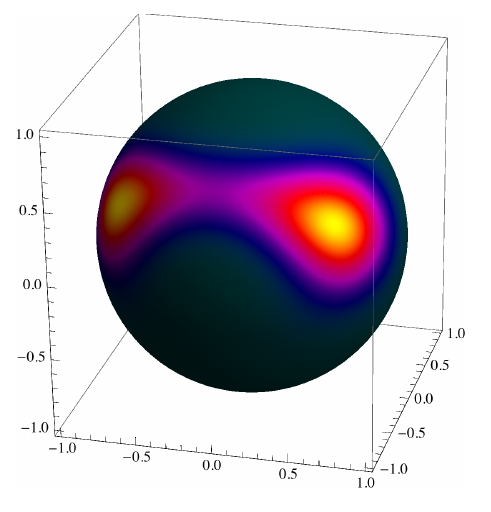}
\includegraphics[width=0.44\columnwidth,clip,trim=30mm 220mm 130mm 25mm]{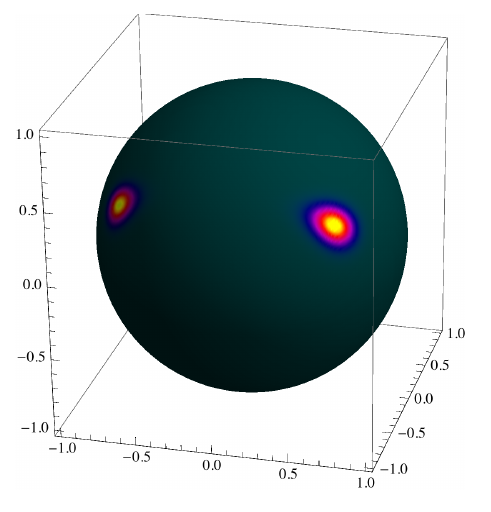}

\caption{{\bf Illustration of
    $\mu_{B^n}$.} \label{fig_estimateillustration} The graphs show
  $\mu_{B^n}$ for measurements performed on $n=20$ and $n=240$ qubits
  (for illustration purposes, we only depict the density $\mu_{B^n}$
  on the surface of the Bloch ball). Half of the qubits have been
  measured in the $z$ direction and half in the $y$ direction with
  relative frequencies of $(0.2, 0.8)$ and $(0.7, 0.3)$,
  respectively. One observes a rapid decrease in the size of the
  bright regions (which are connected by a bright tube inside the
  Bloch ball), which correspond to large values of $\mu_{B^n}$.}
  \end{center}
\end{figure}

\begin{texttheorem}[Reliable Predictions from
  $\mu_{B^n}$] \label{theorem:operational} \label{thm:main} For all
  $B^n$ let $T^{\fail}_{\mu_{B^n}}$ be a POVM element on $\cH^{\otimes
    k}$ such that
$$\int \mu_{B^n}(\sigma) \tr [T^{\fail}_{\mu_{B^n}} \sigma^{\otimes k}] d \sigma
\leq \eps c_{n+k, d}^{-1} \ , $$ where $c_{N, d} = {N+ d^2-1\choose
  d^2-1}$. Then, for any $\rho^{n+k}$,
\begin{align*}
 \big\langle \tr[T^{\fail}_{\mu_{B^n}} \rho_{B^n}^k]\big\rangle_{B^n}
  \leq \eps \ ,
\end{align*}
where $ \big\langle\cdot\big\rangle_{B^n}$ denotes the expectation
taken over all possible measurement outcomes $B^n$ when measuring
$\rho^n$ (i.e., outcome $B^n$ has probability $\tr[B^n \rho^{n}]$).
\end{texttheorem}

As we shall see, the tests are typically chosen such that the integral
over $d \sigma$ decreases exponentially with $n$. The additional
factor $c_{n+k, d}^{-1}$, which is inverse polynomial in $n+k$, plays
therefore only a minor role in the criterion. We also emphasize that
the theorem is valid independently of how the systems $\cS_1, \ldots,
\cS_{n+k}$ have been prepared.  In particular, the (commonly made)
assumption that they all contain identical copies of a single-system
state is not necessary.

The proof of the theorem, together with a slightly more general
formulation, is provided in the \emph{Supplemental Information}.  It
makes crucial use of the following fact, which has also been used in
quantum-cryptographic security proofs: there exists a so-called
\emph{de Finetti state} $\tau^N$, i.e., a convex combination of tensor
products, such that $\rho^{N}\leq c_{N, d} \cdot \tau^N$ holds for all
permutation-invariant states $\rho^{N}$ on $\cH^{\otimes
  N}$~\cite{HayashiApprox, ChKoRe09}.

\bigskip \emph{Confidence Regions.|}A \emph{confidence region} is a
subset of the single-particle state space which is likely to contain
the ``true'' state. In order to formalize this, we consider the
practically relevant case of an experiment that can \emph{in
  principle} be repeated arbitrarily often. Within the above-described
general scenario, this corresponds to the limit where $k$ approaches
infinity while $n$, the number of actual runs of the experiment (whose
data is analyzed), is still finite and may be small.

Since the initial state $\rho^{n+k}$ of all $n+k$ systems can without
loss of generality be assumed to be permutation invariant (see above),
the Quantum de Finetti
Theorem~\cite{HudMoo76,RagWer89,CaFuSc02,CKMR06,Renner07} implies
that, for fixed $n, k' \in \mathbb{N}$, the marginal state
$\rho^{n+k'}$ on $n+k'$ systems is approximated by a mixture of
product states, i.e.,
\begin{align}  \label{eq_iid}
  \rho^{n+k'} = \tr_{k-k'}(\rho^{n+k}) \approx \int P(\sigma) \sigma^{\otimes (n+k')}
  d \sigma \ ,
\end{align}
for some probability density function $P$ and approximation error
proportional to $1/k$. In the limit of large $k$, the marginal state
$\rho^{n+k'}$ is thus fully specified by $P$. We can therefore
equivalently imagine that all systems were prepared in the same
unknown ``true'' state $\sigma$, which is distributed according to $P$
(see Fig.~\ref{fig_iid}). This corresponds to the i.i.d.\ assumption
commonly made in the literature on quantum state tomography, which is
therefore rigorously justified within our general setup.

As before, we assume that tomographic measurements are applied to the
systems $\cS_1, \ldots, \cS_n$, whereas the remaining systems,
$\cS_{n+1}, \ldots, \cS_{n+k'}$, undergo a test (depending on the
output $\mu_{B^n}$ of the data analysis procedure). We may now
consider tests that are passed if and only if the true state $\sigma$
is contained in a given subset $\Gamma^\delta_{\mu_{B^n}}$ of the
state space. The following corollary provides a sufficient criterion
under which the tests are passed, so that $\Gamma^\delta_{\mu_{B^n}}$
are confidence regions.  (Note that the criterion refers to additional
sets $\Gamma_{\mu_{B^n}}$ that are related to the confidence regions
$\Gamma^\delta_{\mu_{B^n}}$; see the \emph{Supplemental Information}
for an illustration.)

\begin{textcorollary}[Confidence Regions from
  $\mu_{B^n}$] \label{cor:main}
  For all $B^n$ let $\Gamma_{\mu_{B^n}}$ be a set of states on $\cH$
  such that
\begin{align}
  \int_{\Gamma_{\mu_{B^n}}} \mu_{B^n}(\sigma) d\sigma \geq 1- \frac{\eps}{2} c_{2n, d}^{-1} \ .
\end{align}
Then, for any $\sigma$, 
$$\mathrm{Prob}_{B^n}[\sigma \in \Gamma^{\delta}_{\mu_{B^n}}]\geq
1-\epsilon \ ,$$ where $\mathrm{Prob}_{B^n}$ refers to the
distribution of the measurement outcomes $B^n$ when measuring
$\sigma^{\otimes n}$ (i.e., outcome $B^n$ has probability $\tr[B^n
\sigma^{\otimes n}]$) and where
\begin{align} \label{eq_Gammadelta} \Gamma^\delta_{\mu_{B^n}}=
  \{\sigma: \exists \sigma' \in \Gamma_{\mu_{B^n}} \text{ with }
  F(\sigma, \sigma')^2 \geq 1- \delta^2 \} \ ,
\end{align}
with $\delta^2 =\frac{2}{n} (\ln \frac{2}{\eps}+ 2 \ln c_{2n, d})$ and
$F(\sigma, \sigma')= \|\sqrt{\sigma} \sqrt{\sigma'}\|_1$ the fidelity.
\end{textcorollary}

\begin{figure}
  \includegraphics[width=0.95\columnwidth,clip,trim = 25mm 198mm 111mm
  18mm]{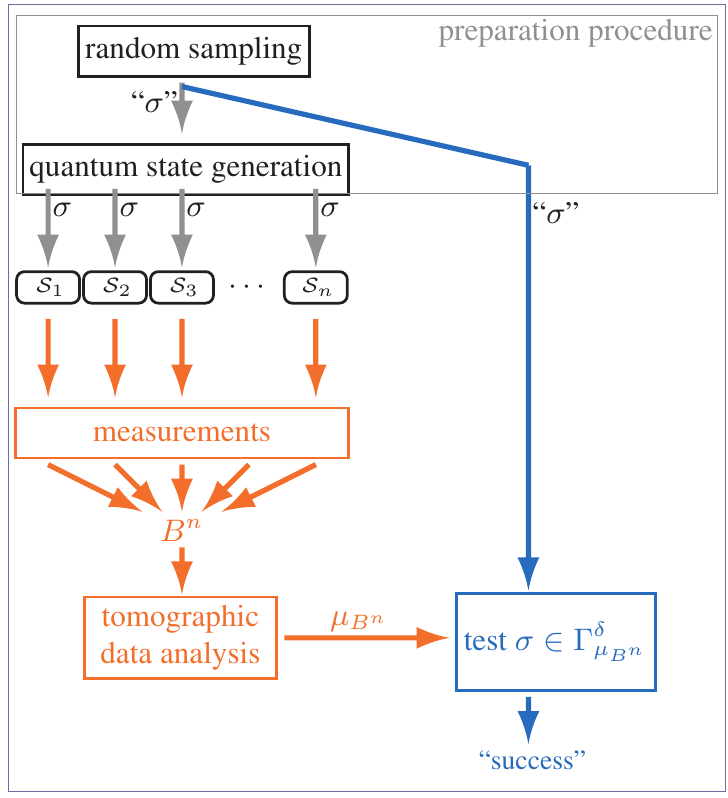}
  \caption{{\bf Tomography of identically prepared
      systems.} \label{fig_iid} This scenario falls into the framework
    depicted in Fig.~\ref{fig_setup}, corresponding to the limit where
    the number of extra systems, $k$, approaches infinity. In this
    case, we can assume without loss of generality that the systems
    have been prepared in a two-step process: first, a description
    ``$\sigma$'' of a single-system state is sampled at random
    (according to some probability density $P$); second, $n$ identical
    systems $\cS_1, \ldots, \cS_n$ are prepared in state $\sigma$. The
    $k$ extra systems of Fig.~\ref{fig_setup} are replaced by a
    classical variable carrying the description ``$\sigma$''.  Given
    only the output of the tomographic data analysis, $\mu_{B^n}$, it
    is possible to decide whether $\sigma$ is (with probability at
    least $1-\eps$) contained in a given set
    $\Gamma_{\mu_{B^n}}^\delta$ (Corollary~\ref{cor:main}). If this is
    the case then $\Gamma_{\mu_{B^n}}^\delta$ is a confidence region
    (with confidence level $1-\eps$).}
\end{figure}

The main idea for the proof of the corollary is to apply the above
theorem to tests (acting on $k' = n$ systems) derived from Holevo's
optimal covariant measurement~\cite{Holevo}. We refer
to the \emph{Supplemental Information} for the technical proof.

Note that $1-\eps$ can be interpreted as the \emph{confidence level}
of the statement that the true state $\sigma$ is contained in the set
$\Gamma^\delta_{\mu_{B^n}}$.  Crucially, the claim is valid for all
$\sigma$. In particular, it is independent of any initial probability
distribution, $P$, according to which $\sigma$ may have been chosen
(see Eq.~\ref{eq_iid}). In other words, the operational interpretation
of the sets $\Gamma_{\mu_{B^n}}^\delta$ as confidence regions does not
depend on any extra assumptions about the preparation procedure or on
the specification of a prior.  In fact, $\sigma$ could even be chosen
``maliciously'', for example in a quantum cryptographic context, where
an adversary may try to pretend that a system has certain properties
(e.g., that its state is entangled while in reality it is not).

Obviously, the assertion that a state $\sigma$ is contained in a
certain set $\Gamma_\mu^{\delta}$ can only be considered a good
approximation of $\sigma$ if the set $\Gamma_\mu^{\delta}$ is
small. This is indeed the case for reasonable choices of the
measurement $\{B^n\}$. For instance, in the practically important case
where each system is measured independently and identically with POVM
$\{E_i\}$, the confidence region is, for generic states,
asymptotically of size proportional to $\frac{1}{\sqrt{n}}$ in the
(semi-)norm on the set of quantum states induced by the POVM: $\|\cdot
\|_{\{E_i\}}=\sum_i |\tr (E_i \cdot) |$ (see \emph{Supplemental
  Information} and~\cite{MattewsWehnerWinter}).

\emph{Conclusion.|}Despite conceptual differences, our technique is
not unrelated to MLE and Bayesian estimation. As mentioned before,
$\mu(\sigma)$ is proportional to the likelihood function and,
therefore, methods to construct confidence regions with our technique
are likely to use adaptations of techniques from MLE. Also,
$\mu_{B^n}(\sigma)d\sigma$ corresponds to the probability measure
obtained from applying Bayes' updating rule to the Hilbert-Schmidt
measure; a fact that implies near optimality~\footnote{More precisely,
  the bound on the parameter $\eps$, which is usually exponentially
  decreasing in the size of the confidence region
  $\Gamma_{\mu_{B^n}}$, is tight up to a polynomial factor.} of our
method in the practically most relevant case of independent
tomographically complete measurements~\cite{TanKom05}.

Recently, another novel approach to quantum state tomography has been
proposed~\cite{CompressedTomo,EfficientTomo}, which yields reliable
error bounds similar to ours. A central difference between this
approach and ours is the level of generality. In~\cite{CompressedTomo,
  EfficientTomo} a specific sequence of measurement operations is
proposed, which is adapted to systems whose states are fairly
pure. Under this condition, the estimate converges fast and, in
addition, can be computed efficiently. In contrast, our method can be
applied to arbitrary measurements (i.e., any tomographic data may be
analyzed). Accordingly, the convergence of the confidence region
depends on the choice of these measurements. However, we do not
propose any specific algorithm for the efficient computation of
confidence regions.

Finally, we refer to the very recent work of
Blume-Kohout~\cite{Blume12} for an excellent discussion of the
notion of confidence regions in quantum state tomography. In
particular, he shows that confidence regions, as considered here, can
be defined via likelihood ratios.

\emph{Acknowledgements.|}We thank Robin Blume-Kohout for useful
comments on earlier versions of this work.  We acknowledge support
from the Swiss National Science Foundation (grants PP00P2-128455 and
200020-135048, and through the National Centre of Competence in
Research `Quantum Science and Technology'), the German Science
Foundation (grant \mbox{CH~843/2-1}), and the European Research
Council (grant 258932).

\bibliographystyle{apsrev}

\newpage

\onecolumngrid
\appendix

\section*{{\Large Supplemental Information}}

{\bf Note: This document contains additional material (Appendices~A,
  B, C, and~D) that is not included in the ``Supplemental
  Information''  accompanying the journal
  version.}

\bigskip

In the first part, we give precise statements and proofs of our
technical results (Section~1), discuss the practically important case
of independent measurements (Section~2) and present further remarks
(Section~3). The second part explains how to
represent states on the symmetric subspace (Appendix~A) as well as
functions on the state space (Appendices~B and~C), and concludes with
examples (Appendix~D). 

\section{1 $\quad$ Statements and Proofs} 

Let $\cH$ be a Hilbert space of finite dimension $d$, i.e., $\cH \cong
\complex^d$. We denote the set density matrices on $\cH$ by $\cS(\cH)$
and the subset of pure states by $\cP(\cH)$. Note that $\cP(\cH)$ can
be identified with $\CP{d-1}$, the complex projective space of
dimension $d-1$. $\CP{d-1}$ carries a natural action of the unitary
group $U(d)$. The Haar measure on $U(d)$ therefore descends to a
measure on $\cP(\cH)$ which is invariant under the action of
$U(d)$. We denote this measure by $d\phi$ and fix the normalisation so
that $\int d\phi =1$. The symmetric subspace $\Sym^n(\cH)$ of
$\cH^{\otimes n}$ is defined as the space of vectors that are
invariant under the action of the symmetric group $S_n$ that permutes
the tensor factors. Since the action of $S_n$ commutes with the action
of the unitary group on $\cH^{\otimes n}$, $U(d)$ acts on
$\Sym^n(\cH)$ as well. We denote the dimension of $\Sym^n(\cH)$ by
$\dim(n, d)$. The following lemma~\cite{HayashiApproxSupp,
  ChKoRe09Supp} is crucial in the derivation of the main results.
\begin{lemma}
\label{lem:postselect}
Let $\rho^n \in \cS(\Sym^n(\complex^d))$. Then 
$$\rho^n\leq \dim(n, d) \int \phi^{\otimes n} d\phi.$$
Furthermore, $\dim(n, d) = {n+d-1 \choose n} \leq (n+1)^{d-1}$.
\end{lemma}

The state defined by the integral on the right hand side is sometimes
called \emph{de Finetti state}.  Note that the corresponding statement
mentioned in the Letter for general permutation-invariant density
operators $\rho^n$ is obtained from this lemma by considering a
purification of $\rho^n$ in the symmetric subspace
(see~\cite{ChKoRe09Supp}).

\begin{proof}
The space $\Sym^n(\complex^{d})$ is irreducible under the action of the unitary group $U(d)$~\cite{FultonHarris91}. The operator $\int \phi^{\otimes n} d\phi$ is supported on $\Sym^n(\complex^d)$ and invariant under the action of $U(d)$. By Schur's lemma we therefore have 
$$\dim(n, d) \int \phi^{\otimes n} d\phi  = \id_{\Sym^n(\complex^d)}.$$
The claim follows since
$$\rho^n \leq \id_{\Sym^n(\complex^d)}$$
holds for any density operator. $\dim(n, d)$ equals ${n+d-1 \choose n}$ and is easily seen to be upper bounded by $(n+1)^{d-1}=\poly(n)$.
\end{proof}

Consider now the measure $d\phi$ on a tensor product space $\cH
\otimes \cK$, where $\cK\cong \cH \cong \complex^d $ and perform the
partial trace operation over system $\cK$. We denote the resulting
measure by $d\sigma$ on $\cS(\cH)$ and note that it may also be
defined as the measure induced by the Hilbert-Schmidt
metric~\cite{ZycSom01}. (In the Letter, we refer to $d \sigma$ as the
Hilbert-Schmidt measure.) For any POVM element $B^n$ on $\cH^{\otimes
  n}$, our data analysis procedure produces a probability distribution
$$\mu_{B^n}(\sigma)d\sigma:= \frac{1}{c_{B^n}} \tr[ B^n \sigma^{\otimes n}]d\sigma,$$
where $c_{B^n}= \int \tr B^n \sigma^{\otimes n} d\sigma= \tr [B^n
\otimes \id_{\cK}^{\otimes n}\cdot \id_{\Sym^n(\cH\otimes \cK)}]$. Let
$\rho^{n+k}$ be a permutation-invariant density operator on
$\cH^{\otimes n+k}$ and $\varrho^{n+k}$ a purification with support on
$\Sym^{n+k}(\cH \otimes \cK)$ (see e.g.~\cite{CKMR06Supp}). That is,
$\varrho^{n+k}$ is pure and $\tr_{\cK^{\otimes
    n}}\varrho^{n+k}=\rho^{n+k}$.  Denote by
$\varrho^{k}_{B^n}=\frac{1}{\tr [B^n \rho^n]}\tr_n [ ( B^n\otimes
\one_{\cK}^{\otimes n} ) \otimes \one_{\cH \otimes \cK}^{\otimes k}
\cdot \varrho^{n+k}] $ the post-measurement state and note that it
appears with probability $\tr [B^n \rho^n]$. Furthermore, we define
$\nu_{B^n}(x):= \frac{1}{c_{B^n}} \tr [B^n \otimes \one_{\cK}^{\otimes
  n} \cdot \proj{x}^{\otimes n}]$ for $x \in \cH \otimes \cK$. The
following theorem holds for any POVM element $T^{\fail}_{\mu_{B^n}}$
on $(\cH \otimes \cK)^{\otimes k}$.

\begin{theorem}[Reliable Predictions from $\mu_{B^n}$] 
If for all ${B^n}$
\begin{align} \label{eq_criterion} \int \nu_{B^n}(x) \tr [T^{\fail}_{\mu_{B^n}} \proj{x}^{\otimes k} ]dx
  \leq   \eps {n+k+ d^2-1\choose d^2-1}^{-1}, 
  \end{align}
then
\begin{align*}
 \big\langle \tr[T^{\fail}_{\mu_{B^n}} \varrho_{B^n}^k]\big\rangle_{B^n}
  \leq \eps \ ,
\end{align*}
where $ \big\langle\cdot\big\rangle_{B^n}$ denotes the expectation
taken over all possible measurement outcomes $B^n$ according to the
probability distribution $\tr[B^n \rho^{n}]$.
\end{theorem}

Note that the probability distribution $\mu_{B^n}(\sigma)d\sigma$ is
obtained from the probability distribution $\nu_{B^n}(x) dx$ by taking
the partial trace over the purifying system~$\cK$. The above theorem
therefore immediately implies the theorem in the Letter, which
corresponds to the specialisation where $T^{\fail}_{\mu_{B^n}}$ acts
only on $\cH^{\otimes k}$.

\begin{proof}
\begin{align*}
\big\langle \tr [T^{\fail}_{\mu_{B^n}} \cdot  & \varrho_{B^n}^k]\big\rangle_{B^n}
 = \sum_{B^n} \tr [(B^n \otimes \one_{\cK}^{\otimes n}) \otimes T^{\fail}_{\mu_{B^n}} \cdot \varrho^{n+k}]\\
& \leq \dim(n+k, d^2) \sum_{B^n} \int \tr [(B^n \otimes \id_{\cK}^{\otimes n} ) \otimes T^{\fail}_{\mu_{B^n}} \cdot \proj{x}^{\otimes n+k}] dx\\
&=  \dim(n+k, d^2) \sum_{B^n} c_{B^n} \int \nu_{B^n}(x)   \tr[T^{\fail}_{\mu_{B^n}} \proj{x}^{\otimes k}]dx\\
& \leq \eps,
\end{align*}
where we used Lemma~\ref{lem:postselect} in the first inequality and
the assumption, \eqref{eq_criterion}, in the second.
\end{proof}

For a subset $\Gamma_{\mu}$ of $\cS(\cH)$, we define 
$$\Gamma^\delta_\mu= \{\sigma: \exists \sigma' \in \Gamma_\mu \text{ with } P(\sigma, \sigma')\leq  \delta \},$$
where $P(\sigma, \sigma')$ is the purified distance defined as
$\sqrt{1-F(\sigma, \sigma')^2}$, where  $F(\sigma, \sigma')= \|
\sqrt{\sigma} \sqrt{\sigma'} \|_1 = \tr \sqrt{\sqrt{\sigma} \sigma' \sqrt{\sigma}}$ is the fidelity. For more details regarding the purified distance as well as its relation to the trace distance see~\cite{dualities}.

\begin{corollary}[Confidence Regions from
  $\mu_{B^n}$]
For all $B^n$, let $\Gamma_{\mu_{B^n}}$ be such that 
\begin{align}\label{eq:gamma}
  \int_{\Gamma_{\mu_{B^n}}} \mu_{B^n}(\sigma) d\sigma \geq 1- \frac{\eps}{2} {2n+d^2-1\choose d^2-1}^{-1} \ 
\end{align}
and let $\delta :=\sqrt{\frac{2}{n} (\ln \frac{2}{\eps}+ 2 \ln {2n+d^2-1\choose d^2-1} )}$. Then for all $\sigma$, 
$$\text{Prob}_{B^n}[\sigma \in \Gamma^{\delta}_{\mu_{B^n}}]\geq 1-\epsilon$$
where the probability is with respect to the measurement outcomes
$B^n$ according to the probability distribution $\tr[B^n
\sigma^{\otimes n}]$. 
\end{corollary}

See Figure~\ref{fig_convergence} for an illustration
of~$\Gamma^{\delta}_{\mu_{B^n}}$ and its relation to $\mu_{B^n}$.

\begin{figure}
\vspace{-0.4cm}
\setlength{\unitlength}{0.11mm} 
\centering
\includegraphics[width=0.7\columnwidth,clip]{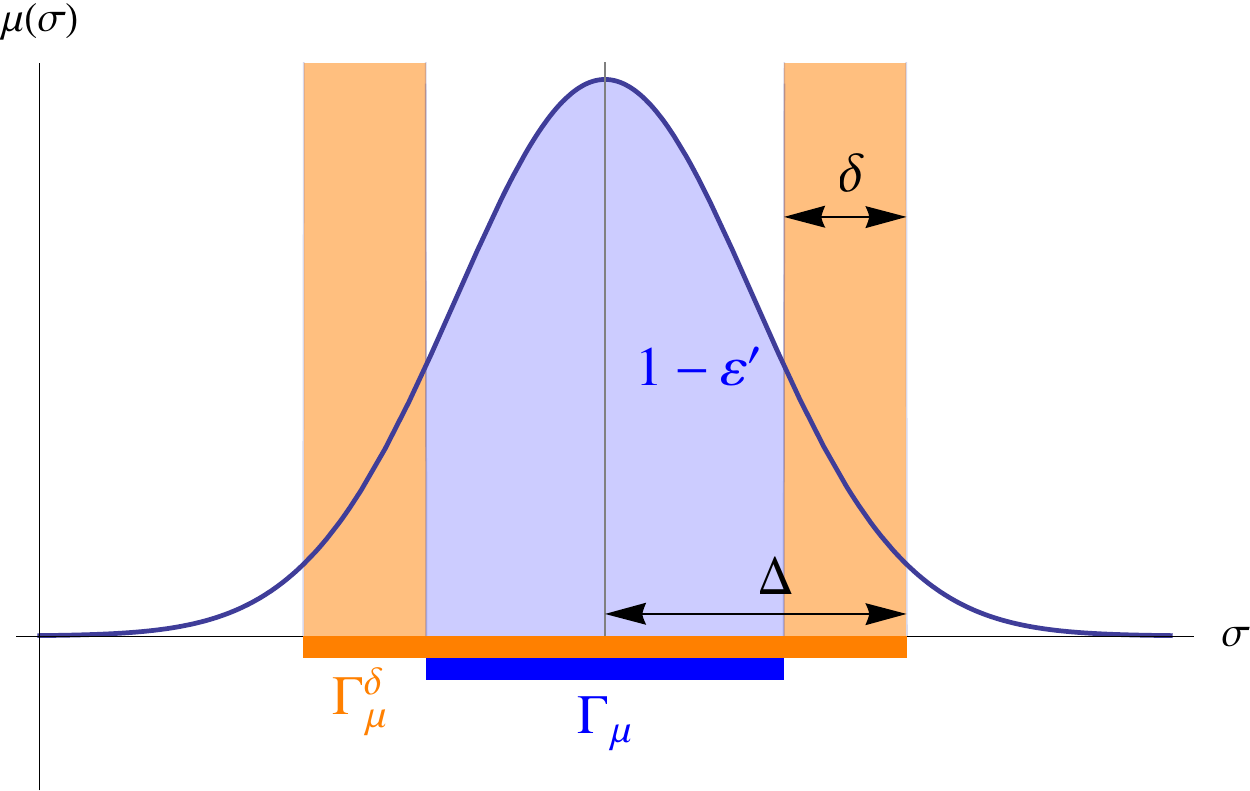}
\caption{{\bf Construction of the confidence region.}  The graph
  schematically illustrates a possible distribution $\mu \equiv
  \mu_{B^n}$ (blue line), a high probability set $\Gamma_{\mu}$ of
  $\mu$ (blue bar), and the set $\Gamma_{\mu}^{\delta}$ which includes
  a $\delta$-region around $\Gamma_{\mu}$ (orange bar). In the
  scenario depicted by Fig.~3 in the Letter, the state $\sigma$ chosen
  by the preparation procedure is with high probability (at least
  $1-\eps$) in $\Gamma_{\mu}^{\delta}$, which can therefore be seen as
  a confidence region.  }
\label{fig_convergence}
\end{figure}

\begin{proof}
The failure probability of the test is given by
\begin{align*}
P^{\fail}(P)& := \int P(\sigma) \sum_{B^n} \tr [B^n \cdot \sigma^{\otimes n}] \Theta_{\overline{\Gamma_{\mu_{B^n}}^\delta}}(\sigma) d\sigma
\end{align*}
where $\Theta_{\overline{\Gamma_{\mu_{B^n}}^\delta}}(\sigma)$ equals one for $\sigma$ in the set $\overline{\Gamma_{\mu_{B^n}}^\delta}$ and zero otherwise ($\overline{A}$ denotes the complement of a subset $A$ of $\cS(\cH)$). This can be rewritten more conveniently as
\begin{align*}
P^{\fail}(P)= \sum_{B^n} c_{B^n} \int_{\overline{\Gamma_{\mu_{B^n}}^\delta}} P(\sigma)  \mu_{B^n}(\sigma) d\sigma.
\end{align*}
Instead of the a priori probability density $P(\sigma)$ we consider a probability density $Q(\psi)$ on the pure states $\cP(\cH \otimes \cK)$, where $\cK \cong \cH$, that gives rise to $P(\sigma)$. That is, $Q(\psi)$ satisfies
$$\int_{A} P(\sigma) d\sigma = \int_{\psi: \tr_{\cK} \psi \in A} Q(\psi) d\psi$$
for all measurable subsets $A$ of $\cS(\cH)$. Note that such a
probability density $Q(\psi)$ always exists as we can purify the state
of the particle with a purifying space $\cK \cong
\cH$. $\mu_{B^n}(\sigma)$ is replaced by
$$\nu_{B^n}(\psi):=\frac{1}{c_{B^n}} \tr [( B^n \otimes \one_{\cK}^{\otimes n}) \cdot  \psi^{\otimes n}].$$ 
Furthermore we consider the following extension of the set
$\Gamma_{\mu}$:
$$\Omega_{\mu} = \{\psi \in \cP(\cH \otimes \cK): \tr_{\cK} \psi  \in \Gamma_\mu \}$$
and of $\Gamma^{\delta}_{\mu}$:
$$\Omega^{\delta}_{\mu} = \{\psi \in \cP(\cH \otimes \cK): \tr_{\cK} \psi  \in \Gamma^\delta_\mu \}.$$
The failure probability of the test can then be expressed as 
\begin{align}\label{eq:fail}
P^{\fail}(P)=\sum_{B^n} c_{B^n} \int_{\overline{\Omega_{\mu_{B^n}}^\delta}} Q(\psi)  \nu_{B^n}(\psi) d\psi.
\end{align}
Let $k' \in \naturals$ be the number of systems that we use to approximate the test by a test procedure that is given by a POVM $\{T^{\fail}, \one- T^{\fail}\}$. Defining 
$$T^{\fail}_{\Omega_\mu^{\delta/2}}:= \dim(k', d^2)\int_{\overline {\Omega_\mu^{\delta/2}}} \phi^{\otimes k'} d\phi. $$
we see that for all 
$\psi \in \overline{\Omega_{\mu}^\delta}$
\begin{align*}
\tr  T^{\fail}_{\Omega_\mu^{\delta/2}} \psi^{\otimes k'}& =1- \dim(k', d^2)\int_{\Omega_\mu^{\delta/2}} \tr[\phi^{\otimes k'} \psi^{\otimes k'} ]d\phi \\
 & \geq 1- \dim(k', d^2)\max_{\phi \in \Omega_\mu^{\delta/2}} F(\tr_{\cK} \phi, \tr_{\cK} \psi)^{k'} \\
& \geq 1- \underbrace{ \dim(k', d^2)e^{-\frac{\delta^2}{2}  k'}}_{=:\eps''}.
\end{align*}
Inserting this estimate into~\eqref{eq:fail} leads to
\begin{align*}
P^{\fail}(P)& \leq \sum_{B^n} c_{B^n} \int_{\overline{\Omega_{\mu_{B^n}}^\delta}} Q(\psi)  \nu_{B^n}(\psi) (\tr [ T^{\fail}_{\Omega^{\delta/2}_{\mu_{B^n}}}\cdot \psi^{\otimes k'}]+\eps'') d\psi.
\end{align*}
We now remove the restriction in the integral, thereby further weakening the estimate and obtain
\begin{align*}
P^{\fail}(P)& \leq \eps'' + \sum_{B^n} c_{B^n}\int Q(\psi)  \nu_{B^n}(\psi) \tr [ T^{\fail}_{\Omega^{\delta/2}_{\mu_{B^n}}}\cdot \psi^{\otimes k'}]d\psi,\\
& =\eps'' + \sum_{B^n} \int Q(\psi) \tr [(B^n \otimes \id_{\cK}^{\otimes n}) \cdot \psi^{\otimes n}]  \tr [ T^{\fail}_{\Omega^{\delta/2}_{\mu_{B^n}}}\cdot \psi^{\otimes k'}] d\psi,\\
& =\eps'' + \sum_{B^n}\tr [(B^n \otimes \id_{\cK}^{\otimes n})\otimes T^{\fail}_{\Omega^{\delta/2}_{\mu_{B^n}}}\cdot ( \int Q(\psi)  \psi^{\otimes n+k'}   d\psi) ].
\end{align*}
Applying Lemma~\ref{lem:postselect} to the state $\int Q(\psi)
\psi^{\otimes n+k'}d \psi$ we can find an upper bound on this quantity
which is independent of the initial distribution $P(\sigma)$ (or
$Q(\psi)$):
\begin{align*}
P^{\fail}(P)& \leq \eps'' + \dim(n+k', d^2) \sum_{B^n}\tr [ (B^n \otimes \id_{\cK}^{\otimes n}) \otimes T^{\fail}_{\Omega^{\delta/2}_{\mu_{B^n}}} \cdot( \int \psi^{\otimes n+k'}   d\psi)],\\
&=\eps''+  \dim(n+k', d^2) \sum_{B^n}c_{B^n}\int_{\Omega_{\mu_{B^n}}} \nu_{B^n}(\psi) \tr  [T^{\fail}_{\Omega^{\delta/2}_{\mu_{B^n}}}\cdot  \psi^{\otimes k'}]   d\psi \\
& \quad +  \dim(n+k', d^2)\sum_{B^n}c_{B^n}\int_{\overline{\Omega_{\mu_{B^n}}}} \nu_{B^n}(\psi) \tr [ T^{\fail}_{\Omega^{\delta/2}_{\mu_{B^n}}}\cdot  \psi^{\otimes k'} ]  d\psi. 
\end{align*}
Since for $\psi \in \Omega_{\mu_{B^n}}$
\begin{align*}
\tr [ T^{\fail}_{\Omega_\mu^{\delta/2}}\cdot \psi^{\otimes k'}]& \leq \dim(k', d^2)\max_{\phi \in\overline{\Omega^{\delta/2}}} F( \tr_{\cK} \phi, \tr_{\cK} \psi)^{k'}\leq  \dim(k', d^2)e^{-\frac{\delta^2}{2}  k'}=\eps''
\end{align*}
and since
$$ \tr [ T^{\fail}_{\Omega^{\delta/2}_{\mu_{B^n}}} \cdot \psi^{\otimes k'} ]\leq 1$$
we find
\begin{align*}
P^{\fail}(P)&\leq \eps''+   \dim(n+k', d^2)\eps'' + \dim(n+k', d^2)\sum_{B^n}c_{B^n}\int_{\overline{\Omega_{\nu_{B^n}}}} \nu_{B^n}(\psi)   d\psi\\
& = \eps''+   \dim(n+k', d^2)\eps'' + \dim(n+k', d^2)\sum_{B^n}c_{B^n}\int_{\overline{\Gamma_{\mu_{B^n}}}} \mu_{B^n}(\sigma)   d\sigma.
\end{align*}
We now set $k'=n$ and use the assumption 
$$  \int_{\Gamma_{\mu_{B^n}} }\mu_{B^n}(\sigma) d\sigma \geq 1- \frac{\eps}{2} \dim(2n, d^2)^{-1}. $$
for all $\mu_{B^n}$. This results in 
\begin{align*}
P^{\fail}(P)&\leq \eps''+   \dim(2n, d^2)\eps'' + \frac{\eps}{2}\leq \dim(2n, d^2)^2 e^{-\frac{\delta^2}{2} n} +\frac{\eps}{2}.
\end{align*}
Choosing $\delta=\sqrt{ \frac{2}{n} (\ln{\frac{2}{\eps}}+ 2 \ln \dim(2n, d^2))}$ ensures that 
\begin{align*}
P^{\fail}(P)&\leq \eps.
\end{align*}
Inserting for $P$ a Dirac delta distribution concludes the proof. 
\end{proof}

\section{2 $\quad$ Independent Measurements}

We now restrict our discussion to the case where $B^n$ is of product
form, i.e.
$$B^n = \prod_{i=1}^r E_i^{\otimes f^{(i)}}$$
where $E_i$ are the elements of a POVM, i.e., $E_i \geq 0$ and
$\sum_iE_i=\id$. $f=(f^{(1)}, \ldots, f^{(r)})$ is the vector
containing the frequencies with which the outcomes occur,
i.e. $f^{(i)} \in \naturals $ and $\sum_{i=1}^n f^{(i)}=n$. We find
$$\mu_{B^n}(\sigma) =\frac{1}{c_{B^n}} \prod_{i=1}^r (\tr E_i \sigma)^{f^{(i)}},$$
where $c_{B^n}=\int \tr B^n \sigma^{\otimes n} d \sigma$. We now want
to investigate the maximum of this function which we assume for
simplicity to be unique, i.e. we assume that the POVM is
tomographically complete. Since the log function is monotonically
increasing, the $\sigma$ which maximises this function can also be
expressed as
\begin{align} \label{eq:MLE} 
\sigma_{\max}:= \mathrm{argmax}_{\sigma} (\sum_i \bar{f}^{(i)} \log  \tr E_i \sigma),
\end{align}
where $\bar{f}^{(i)}=\frac{f^{(i)}}{n}$ are the relative frequencies. This shows that the value at which $\mu_{B^n}$ is maximised coincides with the density matrix that the Maximum Likelihood Estimation method infers since $\sum_i f^{(i)} \log  \tr E_i \sigma$ is the so-called ''log likelihood function''~\cite[Eqs. (2) and (3)]{jezek-hradil}. 

The MLE method is therefore consistent with our method and our work
can be seen as a theoretical justification of it. We emphasize that in
contrast to MLE, our work shows how to compute reliable error
bars. This implies in particular that most of the likely states
(i.e.~the states within the error bars) are not on the boundary of the
state space, even though the maximum might lie on the boundary. Our
method can therefore be seen as a resolution of the ``problem'' that
the states predicted by MLE are unphysical because they lie on the
boundary.

We now want to study the decay of $\mu_{B^n}$ around its maxima. In general this is not easy, as the maxima might lie on the boundary of the set. Useful statements can be made however, when the maxima are in the interior of the set of states. We consider the decay exponent (i.e.~$\frac{-1}{n}$ times the natural logarithm) of the function $\tr B^n \sigma^{\otimes n}$ which is
\begin{align}\label{eq:decay}
- \sum_i \bar{f}^{(i)}\ln \tr E_i \sigma.
\end{align}
We then search for the extreme points of the function 
$$- \sum_i \bar{f}^{(i)}\ln \tr E_i \sigma + c (\tr \sigma -1)$$
where we introduced the Lagrange multiplier $c$ in order to take care of the normalisation of the density matrix. The extreme points are characterised by the equations
$$- \sum_i \bar{f}^{(i)}\frac{E_i}{\tr E_i \sigma}=  c\id $$
$$\tr \sigma =1.$$
Let us for simplicity restrict to the case, where the $E_i$ are linearly independent, then, since the $E_i$ form a POVM, we have 
\begin{align} \label{eq:extreme}\bar{f}^{(i)}= \tr E_i \sigma \end{align}
as a condition for an extreme point. 

In order to carry over these findings to the estimate density, we need to understand the behaviour of the normalisation constant $c_{B^n}=\int \tr B^n \sigma^{\otimes n}d\sigma$.
Since 
$$ \int \tr B^n \sigma^{\otimes n}d\sigma \leq \max_\sigma \tr B^n \sigma^{\otimes n} $$
and by Lemma~\ref{lem:postselect}
$$\int \tr B^n \sigma^{\otimes n}d\sigma \geq \frac{1}{\poly(n)} \max_\sigma \tr B^n \sigma^{\otimes n}$$ we find that $\frac{1}{n}\ln c_{B^n}$ approaches the maximum of \eqref{eq:decay}, $- \sum_i \bar{f}^{(i)}\ln  \bar{f}^{(i)}$, for large $n$. The decay exponent,  $-\frac{1}{n}\ln \mu_{B^n}$, is therefore asymptotically equal to the relative entropy (in units of the natural logarithm)
$$D(\bar{f}\| E(\sigma))= \sum_i \bar{f}^{(i)}(\ln  \bar{f}^{(i)}-\ln  \tr E_i\sigma)$$
where $E(\sigma)=(\tr E_1 \sigma, \ldots, \tr E_r \sigma)$ (this shows in particular that the extreme points are maxima since the relative entropy is nonnegative). Pinsker's inequality 
$$D(\bar{f}\| E(\sigma))\geq \frac{1}{2}\|\bar{f}- E(\sigma)\|_1^2.$$
implies that the error bars around the maxima are therefore given by 
$$\epsilon =  O\left(\frac{1}{\sqrt{n}}\right) $$
in the distance on the set of density matrices induced by the norm~\cite{MattewsWehnerWinterSupp} 
$$\|X\|:= \|E(X)\|_1=\sum_i |\tr E_i X|.$$

\section{3 $\quad$ Remarks} 

We remark that it is possible to adapt our method to the case where
additional information, e.g., about the rank of the state or its
symmetry, is available, and thereby to establish confidence regions
also in these situations. For example, if it is known that the state
$\rho^{n+k}$ is invariant under local actions (on the individual
systems $\cH$) of unitaries from a given set $\cU$, then the integral
in~\eqref{eq:gamma} can be restricted to a subset $\cS$ of states
$\sigma$ on $\cH$ such that $U \sigma U^{\dagger} = \sigma$ for any $U
\in \cU$. Accordingly, the confidence region is given by ${\cS \cap
  \Gamma_{\mu_{B^n}}^{\delta}}$, provided the criterion of the
corollary is satisfied.

We also note that the output of the data analysis, $\mu_{B^n}$, can be
specified using a representation in terms of generalised spherical
harmonics. Data obtained from measurements on $n$ systems specify
exactly the moments of degree less than $n$ of this representation. In
particular, these moments contain all information that is needed for a
later update of $\mu_{B^n}$ based on additional measurement
data (see Appendices). While in the case of i.i.d.\ measurements (with a
finite number of outcomes) this representation of the measurement data
is costly ($O(\text{poly}(n))$ bits) when compared to simply storing
the frequencies ($O(\log n)$ bits), it may be useful in the case of
non-i.i.d.\ measurements or measurements with an unbounded (or even
continuous) set of outcomes, since it does not depend on the number of
outcomes.

\section{A $\quad$ Quasi-Probability Distributions} 

In this section we derive quasi-probability distribution representations for operators on $\Sym^n(\complex^d)$ similar to the P- and Q-representations that are well-known from quantum optics. 

\begin{theorem}[Q-representation]\label{thm:Qrep}
Let $B$ be an operator on $\Sym^{n}(\bbC^{d})$. Then $B$ is uniquely determined by its Q-representation, the function
$$Q_{B}(x)=  \bra{x}^{\otimes n} B \ket{x}^{\otimes n},$$
where $\ket{x}=\sum_i x_i \ket{i}$, $x=(x_1, \ldots, x_d)^T \in \complex^d$ and $\sum_i |x_i|^2=1$. When convenient we will view $Q$ as a function on $\CP{d-1}$.
\end{theorem}

\begin{proof}
We adopt an argument very similar to the one used in
\cite[p.~30]{Perelomov86} in the context of Glauber coherent states. Note that the values
$$\bra{x}^{\otimes n} B \ket{x}^{\otimes n}$$ for $x \in \complex^d$ are determined by $\bra{x}^{\otimes n} B \ket{x}^{\otimes n}$ for $x \in \complex^d$ with $\sum_i |x_i|^2=1$. It therefore suffices to show that the values $ \bra{x}^{\otimes n} B \ket{x}^{\otimes n}$, $x \in \mathbb{C}^{d}$ determine $ \bra{x}^{\otimes n} B \ket{x'}^{\otimes n}$, $x, x' \in \mathbb{C}^{d}$ uniquely. Let $\{ \ket{m} \}$ be the
Gelfand-Zetlin basis for $\Sym^{n} (\bbC^{d})$ (see Appendix~B) and
write
\[
B = \sum_{m,m'} B_{m,m'} \ket{m} \bra{m'}.
\]
The function $ \spr{m'}{x'}^{\otimes n}$ is a polynomial in $x' \in
\bbC^{d}$ and $\bra{x}^{\otimes n}\ket{m}$ is a polynomial in
$\bar{x}$, the complex conjugate of $x$. Defining
\[
\alpha = \frac{\bar{x}+x'}{2},\quad \beta = i \frac{\bar{x}-x'}{2}
\]
we see that
\[
x' = \alpha + i \beta,\quad \bar{x} = \alpha - i \beta
\]
and hence $\bra{x}^{\otimes n} B \ket{x'}^{\otimes n}=\sum_{m, m'}
B_{m, m'}\bra{x}^{\otimes n}\ket{m}\spr{m'}{x'}^{\otimes n}$ is a
polynomial in $\alpha$, $\beta$. Note that every polynomial (in fact
every entire function) is determined by its values for real
parameters
, i.e. by
$\alpha, \beta \in \bbR^d$ in our case. This can be seen by writing the polynomial in the form of a Taylor series (around a real point, e.g. 0). The coefficients in this series are partial derivatives (evaluated at 0), which can be taken in real directions without losing generality and are therefore only dependent on real values of the polynomial. Since for real $\alpha, \beta$
\begin{align*}
\mathrm{Im} (\alpha) &= \frac{1}{2} (-\mathrm{Im} (x) + \mathrm{Im}( x') ) = 0
\\
\mathrm{Im} (\beta) &= \frac{1}{2} ( \mathrm{Re}( x) - \mathrm{Re} (x') ) = 0
\end{align*}
or, in other words $x=x'$, it follows that $\bra{x}^{\otimes n} B
\ket{x'}^{\otimes n}$, as a function of $x, x'$, is wholly
determined by the values $\bra{x}^{\otimes n} B \ket{x}^{\otimes n},
x \in \bbC^d$.
\end{proof}

\begin{theorem}[P-representation]\label{thm:Prep}
Let $B$ be an operator on $\Sym^{n}(\bbC^{d})$.  Then $B$ may be
represented in the form
\[
B = \int_{\complex P^{d-1}} P_{B}(x) \proj{x}^{\otimes n} dx
\]
where
  \[
    P_{B}(x) = \sum_{\ell \in \naturals} \sum_{m} p_{B}(\ell, m) y_{\ell,m}(x)
  \]
  with the second sum extending over the Gelfand-Zetlin basis for the
  irreducible representation of $U(d)$ with highest weight
  $\underbrace{(\ell, 0, \ldots, 0, -\ell)}_{d}$ (see Appendix~B). The
  constants $p_{B}(\ell, m)$ are uniquely determined by $B$ for $\ell
  \leq n$ and are arbitrary otherwise.
\end{theorem}

Two lemmas will be needed in order to prove the theorem.

\begin{lemma}\emph{(\cite[p.~35]{KlaSka85})} 
Let $D$ be the space of operators on $\Sym^n(\bbC^d)$ that can be represented in the form
$$B = \int_{\CP{d-1}} P_B(x) \proj{x}^{\otimes n} dx$$
for some $P_B \in L^2(\CP{d-1})$. Furthermore let $E$ be the space of operators on $\Sym^n(\bbC^d)$ with vanishing Q-representation. Then $D^{\bot} = E$, where $D^{\bot} = \{ A : \tr A B^{\dagger} = 0\quad \forall B \in D \}$.
\label{KlaSkalem}
\end{lemma}

\begin{proof}
If $A \in E$, then for all $B \in D$
\[
\tr A B^{\dagger} = \tr A \int \overline{P_B(x)} \proj{x}^{\otimes n} dx = \int \overline{P_B(x)} \bra{x}^{\otimes n} A \ket{x}^{\otimes n} dx = 0,
\]
hence $A \in D^{\bot}$.
Conversely let $A \in D^{\bot}$, then
\[
\tr A B^{\dagger} = 0 \quad \forall B \in D.
\]
Writing this out results in $\int \overline{P_B(x)} \bra{x^{\otimes n}} A \ket{x^{\otimes n}} dx = 0$, for all functions $P_B(x)$ on $\CP{d-1}$.  This implies
\[
\bra{x^{\otimes n}} A \ket{x^{\otimes n}} =0 \quad \forall x \in \CP{d-1},
\]
since only the identically vanishing function is orthogonal to all square integrable functions on $\CP{d-1}$ (and, in particular, to itself).
\end{proof}

\begin{lemma} \label{lem:onbasis}
The operators $ \int dx \ y_{\ell, m}(x) \proj{x}^{\otimes n}$ are non-vanishing and orthogonal with respect to the Hilbert-Schmidt inner product for $\ell \leq n$ and $m$ a corresponding Gelfand-Zetlin pattern. For $\ell >n$, $ \int dx \ y_{\ell, m}(x) \proj{x}^{\otimes n}=0$.
\end{lemma}
\begin{proof}
We calculate
\begin{align*}
\tr [ \int dx \ y_{\ell, m}(x) & \proj{x}^{\otimes n} \int dz \ \overline{y_{\ell, m}(z)} \proj{z}^{\otimes n}]\\
& =
    \int \int y_{\ell,m}(x) \overline{y_{\ell',m'}(z)}  |\braket{x}{z}|^{2n} dxdz\\
& =
    \int \int y_{\ell,m}(x) \overline{y_{\ell',m'}(z)} \\
   & \qquad \times \frac{1}{\dim(n, d)^2} \sum_{\ell''} \CG{\nu}{\nu^*}{\lambda''}{0}{0}{0}  \sum_{m''} \overline{y_{\ell'', m''}(x)}y_{\ell'', m''}(z)  dxdz\\
& =  \delta_{\ell, \ell'} \delta_{m, m'}  \frac{1}{\dim(n, d)^2}  \CG{\nu}{\nu^*}{\lambda''}{0}{0}{0}  \ , 
\end{align*}
where we have used Corollary~\ref{cor:productexpansion2} (Appendix~B)
in the second equality sign and the orthonormality of the $y$
functions in the third equality sign. Since $ \mult(\ell, n, d)$ is
nonzero for $\ell\leq n $ and vanishes for $\ell>n$ (see
Corollary~\ref{cor:productexpansion2}), this concludes the proof.
\end{proof}

\begin{proof}[Proof of Theorem~\ref{thm:Prep}]
By Theorem~\ref{thm:Qrep}, $\bra{x}^{\otimes n} A \ket{x}^{\otimes n} = 0$ for all $x$ implies $A=0$. Therefore the operator space $E$ from Lemma~\ref{KlaSkalem} contains only the identically vanishing operator. As a consequence, the space $D$ of operators that have a P-representation equals the space of operators on $\Sym^n(\complex^d)$ which proves the first part of the claim.

The Fourier decomposition of $P_B$ (see Appendix~B)
$$P_B(x)= \sum_{\ell, m} p_B(\ell, m) y_{\ell, m}(x)$$
implies the decomposition
$$B= \sum_{\ell, m} p_B(\ell, m)\left( \int dx \ y_{\ell, m}(x) \proj{x}^{\otimes n} \right).$$
From Lemma~\ref{lem:onbasis} we see that the coefficients $p_B(\ell, m)$ are determined by the operator $B$ for $\ell \leq n$ and are arbitrary for $\ell>n$.
\end{proof}

\section{B $\quad$ Spherical Harmonics for Higher Dimensions} 

As we have seen, the functions on $\cS(\cH)$, $\cP(\cH)$ and therefore $\CP{d-1}$ play a central role in the present work. In this section, we perform a Fourier decomposition of the functions defined on $\CP{d-1}$ and derive properties that will enable us to work with these functions very effectively. Whereas this may be considered standard by some readers, we include it for the benefit of completeness. Our construction of an orthonormal basis of functions on $\CP{d-1}$ uses the representation theory of $U(d)$ and its subgroup $U(d-1)\times U(1)$. As general references on representation theory of the unitary group we recommend~\cite{FultonHarris91, CarterSegalMacDonald95}.

A (complex) representation $V$ of a group $G$ is a finite-dimensional complex vector space $V$, equipped with an action of $G$ preserving the group operation. $V$ is irreducible if the only invariant subspaces of $V$ are the empty subspace and $V$ itself. For $H$ a subgroup of a group $G$, let $V \downarrow^{G}_{H}$ denote the restriction of a representation of $G$ to $H$. 

Let $G=U(d)$, $V$ a (holomorphic) representation of $U(d)$, i.e.~a representation whose representing matrices have entries that are holomorphic functions in the variables of $U(d)$, and let $H=T(d)$ be the torus of diagonal matrices in $U(d)$. $V$ decomposes according to
$$V\downarrow^{U(d)}_{T(d)}\cong \bigoplus_w  W_w,$$
where $W_w$ are the isotypic components of the irreducible representations of $T(d)$. Since $T(d)$ is abelian its irreducible representations are one-dimensional. Vectors in $W_w$ are known as \emph{weight vectors} with \emph{weight}   $w=(w_1, \ldots, w_d)$, $w_i \in \bbZ$, that is, for all $\ket{v} \in W_w$: 
$$T(d) \ni t : \ket{v} \mapsto t \ket{v}=t_1^{w_1} \cdots t_d^{w_d} \ket{v},$$
where $t=\diag(t_1, \ldots, t_d)$.
The $w$ with $\dim W_w>0$ are called \emph{weights} of $V$. The lexicographical ordering on the set of weights is the relation $w> w'$ if for the smallest $i$ with $w_i\neq w'_i$, $w_i> w'_i$. It turns out that every irreducible representation $V$ of $U(d)$ has a unique highest weight $\lambda$ satisfying $\dim W_\lambda=1$. $\lambda$ is furthermore \emph{dominant}, i.e.~$\lambda=(\lambda_1, \ldots, \lambda_d)$ satisfies $\lambda_i \geq \lambda_{i+1}$. To every dominant $\lambda$, there also exists an irreducible representation denoted by $V_\lambda$. Two irreducible representations $V_\lambda$ and $V_{\lambda'}$ are equivalent if and only if $\lambda =\lambda'$.

In the case where $G = U(d)$ and $H = U(d-1)$ (embedded as $H\ni h \mapsto \left(\begin{array}{cc} h & 0 \\ 0 & 1 \end{array}\right)\in G$ and using the definition $\ket{i}=(\underbrace{0, \ldots, 0, 1}_{i}, 0, \ldots, 0)^T$ for all $i$) and where the representation of $U(d)$ is irreducible with highest weight $\lambda$ one has the following decomposition, known as the \emph{branching rule} for $U(d)$:
 \begin{align} \label{eq:branching}
    V_\lambda \downarrow^{U(d)}_{U(d-1)}
  \cong
    \bigoplus_{\mu} V_{\mu}
 \end{align}
  where the sum extends over dominant weights $\mu=(\mu_1, \ldots, \mu_{d-1})$ that are \emph{interlaced} by $\lambda$, i.e.~that satisfy
  \begin{equation}
    \lambda_{i+1} \leq \mu_i \leq \lambda_i
  \quad \forall i \in \{1, \ldots, d-1\} \ .
  \label{betcond}
  \end{equation}
Iteratively using the branching rule allows us to define an orthonormal basis of the
representation $V_\lambda$, called \emph{Gelfand-Zetlin basis},
where any basis vector is labeled by a sequence of $d$ diagrams
$\underbrace{\lambda^{(d)}}_{=\lambda}, \underbrace{\lambda^{(d-1)}, \ldots, \lambda^{(1)}}_{=:m}$ such that $\lambda^{(i+1)}$ is interlaced by $\lambda^{(i)}$.  $m$ is called a \emph{Gelfand-Zetlin} pattern for $\lambda$. 
The state with Gelfand-Zetlin pattern $m=( (0^{d-1}), (0^{d-2}), \ldots, (0) )$, where $(0^i)=(\underbrace{0, \ldots, 0}_i)$, will be abbreviated by $m=0$. The corresponding state is $\ket{\lambda, 0}$. 

We denote by $dg$ the volume element of the Haar measure on $U(d)$ with normalisation $\int dg =1$. 
We now consider the Hilbert space of square integrable functions on $U(d)$ with the inner product
$$\int  \overline{\alpha(g)}\beta(g) dg,$$
for two functions $\alpha(g)$ and $\beta(g)$. $L^2(U(d))$ carries a representation of $U(d)\times  U(d)$ when equipped with the action 
$$U(d)\times U(d) \ni (g_1, g_2): \alpha(g) \mapsto \alpha(g_1^{-1} g g_2).$$
Let
\[
  t_{\lambda,m,m'}(g) :=d_\la \bra{\lambda,m} g \ket{\lambda,m'}
\]
be the characteristic (or \emph{representative}) functions, i.e. the matrix elements of the irreducible representations of $U(d)$ (multiplied by $d_\la:=\dim V_\lambda$). Note that these functions are orthonormal with respect to the above defined inner product and --- for fixed $\lambda$ --- span an irreducible representation of $U(d)\times U(d)$ with a pair of highest weights $(\lambda^*, \lambda)$, where the $ \lambda^*$ denotes the highest weight of the $V_\lambda^*$, the representation dual to $V_\lambda$. It is not difficult to check that $\lambda^*=(-\lambda_d, \ldots, -\lambda_1)$. The Peter-Weyl theorem asserts that these functions are dense in the $L^2(U(d))$. Note that one can interpret this theorem as a Fourier theorem on $U(d)$ as it shows that any square integrable function can be expressed as a linear combination of the basis functions $t_{\la, m, m'}$. In the following we want to derive a similar statement for functions on $\CP{d-1}$.

When a group $G$ acts transitively on a set $X$ one can identify $X$
with the set $G \slash H$ of left-cosets of the stabiliser group $H$ of
a point $x_0 \in X$, i.e., the group $H := \{g \in G: g x_0 = x_0\}$ and
the isomorphism $G \slash H \rightarrow X$ is $g H \mapsto g x_0$.
[See \cite[p.~59]{CarterSegalMacDonald95}]. In the following we consider the transitive action of $U(d)$ on $\CP{d-1}$ and let $x_{0}$ be the point with homogeneous coordinates $[0: \cdots: 0: 1]$.  Then $H= U(d-1)\times U(1)$ and $\CP{d-1} \cong U(d) / [ U(d-1)\times U(1)]$, \cite[p.~278]{VilKlimVol2}. 

We will show below that the vectors $\ket{\lambda, 0}$ for 
\begin{align}
\label{eq:special}
\lambda = (\ell, 0, \ldots, 0, -\ell)
\end{align}
 are exactly the ones being stabilized by $U(d-1)\times U(1)$. For such $\lambda$, we can therefore define the functions $y_{\ell, m}$ on $\CP{d-1}$ by
\begin{equation}
 y_{\ell, m}(x) := t_{\lambda, m, 0}(g)
\end{equation}
for $g \in x$. The index $\ell$ is sometimes called a
\emph{moment}. Since the measure $dg$ on $U(d)$ descends to a measure
$dx$ on $\CP{d-1}$, these functions are also square integrable and
orthonormal with respect to the standard inner product. The following
theorem, the main statement of this section, asserts that these
functions span $L^2(\CP{d-1})$ densely.

\begin{theorem} \label{thm:decomposition}
  Let $\mu \in L^2(\CP{d-1})$. Then
  \[
    \mu(x) = \sum_{\ell \in \naturals} \sum_{m} \mu(\ell, m) y_{\ell,m}(x)
  \]
  where the second sum ranges over Gelfand-Zetlin pattern $m$ associated to the irreducible representation
  $(\ell, 0, \cdots, 0, -\ell)$ of $U(d)$. The constants $\mu(\ell, m)$ are square summable.
\end{theorem}
The proof is based on the following extension of the Peter-Weyl theorem. Define
\[
  V_\lambda^H
:=
  \{v \in V_\lambda: \, h \ket{v} = \ket{v}, \, \forall h \in H\}
\]
as the $H$ invariant subspace of $V_{\lambda}$. 
\begin{theorem}[Peter-Weyl theorem] \label{thm:PeterWeyl}
  \[
    L^2(G/H) \cong \bigoplus_\lambda^\wedge V_{\lambda}^* \otimes {V_{\lambda}}^H
  \]
  where ${\displaystyle \bigoplus^\wedge}$ is the completion of the
  direct sum.  A basis for $V_{\lambda}^* \otimes V_{\lambda}$ is given by $t_{\lambda,m,m'}$.
\end{theorem}
\begin{proof}
See e.g.~\cite[Corollary~9.14]{CarterSegalMacDonald95}.
\end{proof}

The following lemma characterises the components in the direct sum in terms of the functions $y_{\ell, m}$.
\begin{lemma} \label{lem:Uinvariant} We have
  \[
  V_\lambda^* \otimes  {V_\lambda}^{U(d-1) \times U(1)}
  =
    \begin{cases}
     \spanv \{ y_{\ell, m}\}& \lambda = (\ell, 0, 0, \ldots, 0, -\ell) \\
      0 & \text{otherwise} \ .
    \end{cases}
  \]
\end{lemma}

\begin{proof}
Since $t_{\lambda, m, m'}(g)=d_\la \bra{\lambda, m} g \ket{\lambda, m'}$, it suffices to show that the vectors $\ket{\lambda, 0}$ with $\lambda$ as in the statement are exactly the vectors fixed by $U(d-1)\times U(1)$.  The claim is therefore equivalent to 
  \[
 {V_\lambda}^{U(d-1) \times U(1)}
  =
    \begin{cases}
     \spanv \{ \ket{\lambda, 0}\}& \lambda = (\ell, 0, 0, \ldots, 0, -\ell) \\
      0 & \text{otherwise} \ .
    \end{cases}
  \]
  
It follows from the branching rule and simple counting of degrees of the polynomials that 
$$ V_\lambda \downarrow^{U(d)}_{U(d-1)\times U(1)}\cong \bigoplus_\mu V_\mu \otimes V_{|\lambda|-|\mu|} $$
where the sum extends over $\mu$ that are interlaced by $\lambda$ and $V_{|\lambda|-|\mu|}$ is the one-dimensional representation of $U(1)$ with weight $|\lambda|-|\mu|$. Since $V_\lambda^H = (V_\lambda \downarrow_H^G)^H$, we find
$$V_\lambda^{U(d-1)\times U(1)}= \left( V_\lambda \downarrow^{U(d)}_{U(d-1)\times U(1)}\right)^{U(d-1)\times U(1)}\cong \bigoplus_\mu V_\mu^{U(d-1)} \otimes V_{|\lambda|-|\mu|}^{U(1)}. $$
$V_\mu^{U(d-1)}$ is exactly nonzero when $V_\mu$ is the trivial representation, i.e.~$\mu=(0^{d-1})$. Likewise, $V_{|\lambda|-|\mu|}^{U(1)}$ is non-vanishing only when $V_{|\lambda|-|\mu|}$ is the trivial representation of $U(1)$, i.e.~$|\lambda|-|\mu|=0$. Note that $(0)^{d-1}$ interlaces $\lambda$ only when $\lambda=(\lambda_1, 0,  \ldots, 0, \lambda_{d})$ and that  $|\lambda|=|\mu|$ furthermore implies $\lambda_1+\lambda_d=0$. Setting $\la_1=\ell$ completes the proof.
\end{proof}

\begin{proof}[Proof of Theorem~\ref{thm:decomposition}]
We apply Theorem~\ref{thm:PeterWeyl} to $G=U(d)$ and $H=U(d-1)\times U(1)$. Recalling that in this case $G/H\cong \CP{d-1}$ the left hand side becomes $L^2(\CP{d-1})$. According to Lemma~\ref{lem:Uinvariant}, the right hand side equals the space spanned by the functions $y_{\ell, m}$. This concludes the proof.
\end{proof}

We now want to relate the multiplication of the $t$ functions to the Clebsch-Gordan coefficients of $U(d)$. In general we have the decomposition
$$V_\lambda \otimes V_{\lambda'}\downarrow^{U(d)\times U(d)}_{U(d)}\cong \bigoplus_{\lambda''} \complex^{c_{\lambda, \lambda'}^{\lambda''}} \otimes V_{\lambda''},$$
where $U(d)$ is embedded diagonally into $U(d)\times U(d)$, i.e.~$U(d) \ni g \mapsto g\times g \in U(d)\times U(d)$. 
The multiplicities $c_{\la, \la'}^{\la''}$ are the well-known Littlewood-Richardson coefficients (see e.g.~\cite{fulton97}). In terms of a basis transform this isomorphism reads
$$\ket{\la, m} \ket{\la', m'}=\sum_{\la'', m'', r} \braket{\la, \la', r, \la'', m''}{\la, m}\ket{\la', m'}\ket{\la, \la', r, \la'', m''}$$
with the $U(d)$ Clebsch-Gordan coefficients $\braket{\la, \la', r, \la'', m''}{\la, m}\ket{\la', m'}$, where $r$ counts the different copies of $V_{\la''}$. The following lemma relates the product of two functions $t_{\lambda, m, 0}$ and $t_{\lambda', m',0}$ to the $U(d)$ Clebsch-Gordan coefficients. More generally, such a formula can be derived for the product of $t_{\lambda, m, \tilde{m}}$ and $t_{\lambda, m', \tilde{m}'}$ functions (see~\cite[Chapter 18.2.1]{VilKlimVol3}).
\begin{lemma}\label{lem:productexpansion}
\begin{align*}
t_{\la, m, 0}(g)t_{\la', m', 0}(g) =\sum_{\la'', m''} \CG{\la}{\la'}{\la''}{m}{m'}{m''}t_{\la'', m'', 0}(g)
\end{align*}
where 
\begin{align}\label{eq:symbol}
&\CG{\la}{\la'}{\la''}{m}{m'}{m''}  := \int t_{\la, m, 0}(g)t_{\la', m', 0}(g)\overline{t_{\la'', m'', 0}(g)}dg\\
&\qquad =\frac{d_\la d_{\la'}}{d_{\la''}}\left(\sum_{r}\braket{\la, \la', r, \la'', 0}{\la, 0}\ket{\la', 0}\bra{\la, m}\braket{\la', m'}{\la, \la', r, \la'', m''}  \right). \nonumber
\end{align}
\end{lemma}
\begin{proof}
\begin{align*}
t_{\la, m, 0}(g)t_{\la', m', 0}(g)	& = d_\la d_{\la'} \bra{\la, m} g \ket{\la, 0} \bra{\la', m'} g \ket{\la', 0} \\
					& = d_\la d_{\la'} (\bra{\la, m} \bra{\la', m'}) g (\ket{\la, 0} \ket{\la', 0}) \\
					& =d_\la d_{\la'}\sum_{\la'', m'', r}\braket{\la, \la', r, \la'', 0}{\la, 0}\ket{\la',0}\\
					& \qquad \qquad  \times \bra{\la, m}\braket{\la', m'}{\la, \la', r, \la'', m''}  \bra{\la'', m''} g \ket{\la'', 0} \\
					& =\sum_{\la'', m''} \CG{\la}{\la'}{\la''}{m}{m'}{m''}t_{\la'', m'', 0}(g)
\end{align*}
\end{proof}

This leads to an important product formula which we use, for instance, in the update rule.
\begin{corollary}\label{cor:productexpansion} 
\begin{align*}
y_{\ell, m}(x)y_{\ell', m'}(x) =\sum_{\ell''}^{\ell +\ell'} \sum_{m''} \CGG{\ell}{\ell'}{\ell''}{m}{m'}{m''}y_{\ell'', m''}(x)
\end{align*}
where  $ \CGG{\ell}{\ell'}{\ell''}{m}{m'}{m''}:= \CG{\la}{\la'}{\la''}{m}{m'}{m''} $
for $\la=(\ell, 0, \ldots, 0, -\ell)$ and similarly for $\la'$ and $\la''$ (see~\eqref{eq:symbol}). 
\end{corollary}

\begin{corollary} \label{cor:productexpansion2}
\begin{align}
\dim(n, d)|\braket{d}{x}|^{2n}=\frac{1}{\dim(n, d)}\sum_{\ell} \CG{\nu}{\nu^*}{\la}{0}{0}{0}  y_{\ell, 0}(x).
\end{align}
where
\begin{align}\label{eq:CGmult}
\CG{\nu}{\nu^*}{\la}{0}{0}{0}  =\frac{\dim(n, d)^2}{d_{\la}}|\braket{\nu, \nu^*, \la, 0}{\nu, 0}\ket{\nu^*, 0}|^2  
\end{align}
for $\la=(\ell, 0, \ldots, 0, -\ell)$ with $\ell\leq n$. Furthermore, $\CG{\nu}{\nu^*}{\la}{0}{0}{0} \neq 0$ for $\ell \leq n$ and vanishes for $\ell >n$.
\end{corollary}
\begin{proof}
Note that $\dim(n, d)=d_\nu=d_{\nu^*}$. Then
\begin{align*}
d_{\nu}^2|\braket{d}{x}|^{2n}& = d_{\nu}^2\bra{d}^{\otimes n} g^{\otimes n} \ket{d}^{\otimes n} \overline{\bra{d}^{\otimes n} g^{\otimes n} \ket{d}^{\otimes n}}\\
& = d_\nu^2 \bra{\nu, 0} g \ket{\nu, 0}\overline{\bra{\nu, 0} g \ket{\nu, 0}}\\
& = t_{\nu, 0, 0} (x) \overline{t_{\nu, 0, 0} (x)} = t_{\nu, 0, 0} (x) t_{\nu^*, 0, 0} (x) 
\end{align*}
since $ \ket{d}^{\otimes n}$ is a weight vector in $\nu=(n, 0, \ldots,
0)$ that is invariant with respect to the subgroup $U(d-1)$ (embedded
into $U(d)$ by inclusion into the top left corner), and therefore has
a Gelfand-Zetlin pattern $m=0$. The invariance with respect $U(d-1)$
follows from the tensor production action of the group $U(d)$ as well
as the fact that the stabilizer of $\ket{d}$ contains
$U(d-1)$. \eqref{eq:CGmult} holds since the Littlewood-Richardson
coefficient $c_{\nu, \nu^*}^{\lambda}$ equals one for $\ell\leq n$ and
vanishes for larger values of $\ell$~\cite{fulton97}. This implies in
particular that $ \CG{\nu}{\nu^*}{\la}{0}{0}{0}$ vanishes for
$\ell>n$. In order to see that $ \CG{\nu}{\nu^*}{\la}{0}{0}{0}$ does
not vanish for smaller values of $\ell$, note that the projection of
$\ket{\nu, 0}\ket{\nu^*, 0}$ onto the irreducible representation
$\lambda=(\ell, 0, \ldots, 0, -\ell)$ is given by
\begin{align*}
P_\lambda \ket{\nu, 0}\ket{\nu^*, 0}&=\sum_m (\braket{\nu, \nu^*, \lambda, m}{\nu, 0}\ket{\nu^*, 0})\ket{\nu, \nu^*, \lambda, m}\\
&= (\braket{\nu, \nu^*, \lambda, 0}{\nu, 0}\ket{\nu^*, 0})\ket{\nu, \nu^*, \lambda, 0}
\end{align*}
where we used the fact that the sum can only contain $H$ invariant vectors. The claim follows since we know that this projection cannot vanish, as the Littlewood-Richardson coefficient is nonzero for all $\ell \leq n$.
\end{proof}

\begin{lemma}\label{lem:multiply}
For $g \in zH$, where $H=U(d-1)\times U(1)$, we have
\begin{align*}
y_{\ell, 0}(g^\dagger x)= \sum_m \overline{y_{\ell, m}(z)} y_{\ell, m}(x)
\end{align*}
\end{lemma}
\begin{proof}
\begin{align*}
\int y_{\ell, 0}(g^\dagger x) \overline{y_{\ell, m}(x)}dx & = d_\la d_{\la'}\int \bra{\la, 0} g^{\dagger} \tilde{g}\ket{\la, 0} \bra{\la', 0} \tilde{g}^\dagger \ket{\la', m}d\tilde{g}\\
& = d_\la \delta_{\la, \la'} \bra{\la, 0} g^\dagger \ket{\la, m}= \delta_{\la, \la'} \overline{y_{\ell, m}(z)}
\end{align*}
\end{proof}

The next corollary generalises the decomposition in Corollary~\ref{cor:productexpansion2}. It is needed in some technical aspects of the paper as well as the examples.
\begin{corollary} \label{cor:productexpansion3}
\begin{align*}
\dim(n, d) |\braket{z}{x}|^{2n}& =\frac{1}{\dim(n, d)} \sum_{\ell}^n  \CG{\nu}{\nu^*}{\la}{0}{0}{0} \sum_m \overline{y_{\ell, m}(z)}y_{\ell, m}(x), 
\end{align*}
where $\nu=(n, 0, \ldots, 0)$ and $\la=(\ell, 0, \ldots, 0, -\ell)$. The coefficients are defined in~\eqref{eq:symbol}.
\end{corollary}
\begin{proof}
By Corollary~\ref{cor:productexpansion2} and~\eqref{eq:symbol} we have
\begin{align}
d_\nu |\braket{z}{x}|^{2n}& = \frac{1}{d_{\nu}}\sum_{\ell}^n  \CG{\nu}{\nu^*}{\la}{0}{0}{0}  y_{\ell, 0}(g^\dagger x) \nonumber \\
& =\frac{1}{d_{\nu}} \sum_{\ell}^n    \CG{\nu}{\nu^*}{\la}{0}{0}{0} \sum_m \overline{y_{\ell, m}(z)}y_{\ell, m}(x) 
\end{align}
where $g \in zH$ ($\ket{x}=g \ket{d}$)and where we used Lemma~\ref{lem:multiply} in the last equation.
\end{proof}

\section{C $\quad$ Recovering the Spherical Harmonics on the Bloch
  Sphere}

In the following we restrict our attention to the special case $d=2$. The complex projective space $\CP{1}$ can be viewed as the sphere $S^2$ with $x \in \CP{1}$ being represented as a point on the sphere parameterised by angles $\theta \in [0, \pi]$ and $\phi \in [0, 2\pi)$. The measure $dx$ turns into $\frac{1}{4\pi} \sin \theta d\theta d\phi$. As a unitary representative  for $x$, $g \in xH$, we choose
\begin{align}
\label{eq:g}
g=e^{i \frac{\phi}{2} \sigma_z} e^{i \frac{\theta}{2} \sigma_x} =\begin{pmatrix} e^{i\frac{\phi}{2}} & 0 \\ 0 &e^{-i\frac{\phi}{2}} \end{pmatrix} \begin{pmatrix} \cos \frac{\theta}{2} & i \sin \frac{\theta}{2} \\  i \sin \frac{\theta}{2} &\cos \frac{\theta}{2} \end{pmatrix} .
\end{align}
This implies 
$\ket{x}=g \ket{2}=i e^{i\frac{\phi}{2}}\sin \frac{\theta}{2} \ket{1} + e^{-i\frac{\phi}{2}}\cos \frac{\theta}{2} \ket{2} $.
where we used 
$$ \ket{1}=\begin{pmatrix} 1 \\ 0 \end{pmatrix}\qquad \ket{2}=\begin{pmatrix} 0 \\ 1 \end{pmatrix}  $$
We may think of $\theta=0$ as the south pole and $\theta=\pi$ as the north pole of the sphere (when the $z$ direction is the rotation axis of the earth).

It is then natural to expect that the $y_{\ell, m}$ are related to the ordinary spherical harmonics on the sphere. The next lemma provides us with the precise dependence. Thereafter we will find a formula for the coefficients $\CG{\nu}{\nu^*}{\la}{*}{*}{*}$ that govern the multiplication of two functions. 

Before we start, note that for $d=2$ the possible Gelfand-Zetlin
patterns $m$ for $\lambda=(\ell, -\ell)$ lie in the interval $-\ell
\leq m \leq \ell$ and that the spin projection in $z$-direction
(i.e. the eigenvalue of the (Lie algebra) representation of the
operator $\frac{1}{2}\sigma_z$) of the state $\ket{\lambda, m}$ equals
$m$ (cf.~Lemma~\ref{lem:SU2}).

\begin{lemma}
For $d=2$ 
\begin{align*} 
y_{\ell, m}(x)&=(-i)^m  \sqrt{2 \ell+1} \sqrt{\frac{(\ell-m)!}{(\ell+m)!}}P^m_\ell(\cos \theta) e^{im\phi}=(-i)^m\sqrt{4\pi} Y_{\ell}^{m}(\theta, \phi)
\end{align*}
where $P_\ell^m$ are the associated Legendre polynomials and 
$$Y_{\ell}^{m}(\theta, \phi):= \sqrt{\frac{2 \ell+1}{4\pi}} \sqrt{\frac{(\ell-m)!}{(\ell+m)!}}P^m_\ell(\cos \theta) e^{im\phi}$$
are the spherical harmonics~\footnote{Note that we are using a standard convention also used in Mathematica.}.
\end{lemma}
\begin{proof}
By \cite[eq (1) in Chapter 6.3.1]{VilKlimVol1}
$$\bra{\ell, m} g \ket{\ell, 0}= {\bf t}_{-m, 0}(g).$$
Using~\eqref{eq:g} and \cite[eq (2) \& (4) in Chapter 6.3.3 and eq.~(3) in Chapter 6.3.7.]{VilKlimVol1} we see that for positive $m$
$$e^{i m \phi} (-i)^m\sqrt{ \frac{(\ell-m)!}{(\ell+m)!}}P^m_\ell(\cos \theta)$$
holds, where $P^m_\ell$ denote the associated Legendre polynomials. The same formula can be seen to hold for negative $m$ by use of~\cite[eq (2') in Chapter 6.3.6 and eq (3') in Chapter 6.3.7]{VilKlimVol1}.

\end{proof}

We will now find formulae for the $\CG{\nu}{\nu^*}{\la}{*}{*}{*}$ by relating them to the Clebsch-Gordan coefficients of $SU(2)$ for which closed formulae are known. We start by relating the Clebsch-Gordan coefficients of $SU(2)$ and $U(2)$.
\begin{lemma} \label{lem:SU2}
Let $d=2$. If $c_{\la, \la'}^{\la''} \neq 0$ then
$$\braket{\la, \la', \la'', m''}{\la, m}\ket{\la', m'}_{U(2) }= \braket{L, L', L'', M''}{L, M}\ket{L', M'}_{SU(2)} $$
where $$\la=(\la_1, \la_2) \qquad \qquad L=\frac{\la_1-\la_2}{2} \qquad \qquad M=m-\frac{\la_1+\la_2}{2}$$
and likewise for the primed variables.
\end{lemma}
\begin{proof}
If $V_{\la''} \subset V_{\la} \otimes V_{\la'}$, then 
$$V_{\la''}\downarrow^{U(2)}_{SU(2)} \subset V_{\la}\downarrow^{U(2)}_{SU(2)} \otimes V_{\la'}\downarrow^{U(2)}_{SU(2)},$$
Since $V_{\la}\downarrow^{U(2)}_{SU(2)}$ is equivalent to a $\mathrm{spin-} L$ representation (with $L$ as in the claim) we can obtain the Clebsch-Gordan coefficients for $U(2)$ from those of $SU(2)$. This works as follows. The mapping of the basis state in the irreducible representation $\la= (\la_1, \la_2)$ with Gelfand-Zetlin pattern $(m)$ is
$$\ket{\la, m}_{U(2)} \rightarrow \ket{L, M}_{SU(2)},$$
where  $L$ and $M$ are defined as in the statement of the claim since the weight of a Gelfand-Zetlin pattern $m$ in a representation $\la$ equals $(w_1, w_2)=(m, \la_1+\la_2-m)$ and the spin projection $M$ along the $z$-direction equals $\frac{w_1-w_2}{2}$.
This concludes the proof. 
\end{proof}

\begin{lemma} \label{lem:qubitCG}
Let $d=2$ and $\la=(\ell, -\ell)$, $\nu=(n, 0)$ and $\ell \leq n$. If $\ell$ is even, then 
\begin{align} \label{eq:deltafourier}
\braket{\nu, \nu^*, \la, 0}{\nu, 0}\ket{\nu^*, 0}_{U(2) }= \frac{\sqrt{2\ell+1}}{\sqrt{n+\ell+1}}\frac{n!}{\sqrt{(n-\ell)!(n+\ell)!}} 
\end{align}
and zero otherwise. If $n$ and $\ell$ are even
$$\braket{\nu, \nu^*, \la, 0}{\nu, \frac{n}{2}}\ket{\nu^*, -\frac{n}{2}}_{U(2)}=\frac{(-1)^{\frac{n-\ell}{2}} }{\sqrt{n+\ell+1}}\frac{ \frac{n+\ell}{2}!\ell!\sqrt{(n-\ell)!}}{\frac{n-\ell}{2}!\frac{\ell}{2}!^2\sqrt{(n+\ell)!}}$$ and zero otherwise.
\end{lemma}
\begin{proof}
By Lemma~\ref{lem:SU2}
$$\braket{\nu, \nu^*, \la, 0}{\nu, 0}\ket{\nu^*, 0}=\braket{\frac{n}{2}, \frac{n}{2}, \ell, 0}{\frac{n}{2}, -\frac{n}{2}}\ket{\frac{n}{2}, \frac{n}{2}}.$$ 
Using the formula~\cite[eq (4) in Chapter 8.2.4]{VilKlimVol1}
$$\braket{\ell, \ell', \ell'', \ell-\ell'}{\ell, \ell}\ket{\ell', -\ell'}_{SU(2)}=\sqrt{\frac{(2\ell''+1)(2\ell)!(2\ell')!}{(\ell+\ell'-\ell'')!(\ell+\ell'+\ell''+1)!}}$$
we find
$$\braket{\frac{n}{2}, \frac{n}{2}, \ell, 0}{\frac{n}{2}, \frac{n}{2}}\ket{\frac{n}{2}, -\frac{n}{2}}= \frac{\sqrt{2\ell+1}}{\sqrt{n+\ell+1}}\frac{n!}{\sqrt{(n-\ell)!(n+\ell)!}}$$
This concludes the proof of~\eqref{eq:deltafourier}.
By Lemma~\ref{lem:SU2} 
$$\braket{\nu, \nu^*, \la, 0}{\nu, \frac{n}{2}}\ket{\nu^*, -\frac{n}{2}}_{U(2)}=\braket{\frac{n}{2}, \frac{n}{2}, \ell, 0}{\frac{n}{2}, 0}\ket{\frac{n}{2}, 0}_{SU(2)}.$$ 
Using the formula~\cite[eq (8) in Chapter 8.2.6]{VilKlimVol1}
\begin{align}
\braket{\ell, \ell', \ell'', 0}{\ell, 0}\ket{\ell', 0} =\frac{(-1)^{g-\ell''}g!\Delta(\ell, \ell', \ell'')\sqrt{2\ell''+1}}{(g-\ell)!(g-\ell')!(g-\ell'')!}
\end{align}
where $2g:=\ell+\ell'+\ell''$ is even (the coefficient vanishes for odd $2g$) and (see~\cite[eq (3) in Chapter 8.1.3]{VilKlimVol1}
$$\Delta(\ell, \ell', \ell'')= \sqrt{\frac{(\ell+\ell'-\ell'')!(\ell-\ell'+\ell'')!(\ell'-\ell+\ell'')!}{(\ell+\ell'+\ell''+1)!}}.$$ If $\ell=\ell'$ the coefficient vanishes unless $\ell''$ is even. 
If $n$ and $\ell$ are even we find
$$\braket{\frac{n}{2}, \frac{n}{2}, \ell, 0}{\frac{n}{2}, 0}\ket{\frac{n}{2}, 0}_{SU(2)}= \frac{(-1)^{\frac{n-\ell}{2}} }{\sqrt{n+\ell+1}}\frac{ \frac{n+\ell}{2}!\ell!\sqrt{(n-\ell)!}}{\frac{n-\ell}{2}!\frac{\ell}{2}!^2\sqrt{(n+\ell)!}}.$$ Otherwise the coefficient vanishes.
\end{proof}

\begin{corollary}\label{cor:clebsch-gordon-estimate2}
For $\la=(\ell, -\ell)$ and $\nu=(n, 0)$ and $\ell$ and $n$ even we have 
\begin{align*} 
\frac{1}{d_\nu}\CG{\nu}{\nu^*}{\la}{\frac{n}{2}}{-\frac{n}{2}}{0}&=(-1)^{n-\frac{\ell}{2}} (\frac{1}{2})^{\ell} \binom{\ell}{\frac{\ell}{2}} \prod_{i=1}^{\ell}\left( 1-\frac{\ell-i}{n+2+i}\right)
\end{align*}
\begin{align*} 
\frac{1}{d_\nu}\CG{\nu}{\nu^*}{\la}{0}{0}{0}&=\frac{n!(n+1)!}{(n-\ell)!(n+\ell+1)!}
\end{align*}
\end{corollary}

\begin{proof}
By Lemma~\ref{lem:qubitCG} we have
\begin{align*}
(-1)^{\frac{n-\ell}{2}} & 
\frac{\dim V_\nu}{\dim V_\lambda}\braket{\nu, \nu^*, \la, 0}{\nu, 0}\ket{\nu^*, 0}\bra{\nu, \frac{n}{2}}\braket{\nu^*, -\frac{n}{2}}{\nu, \nu^*, \la, 0}\\
&=\frac{(\frac{n+\ell}{2})!\ell! (n+1)!}{(\frac{n-\ell}{2})!(\frac{\ell}{2})!^2 (n+1+\ell)!}\\
&=\frac{ (\frac{n}{2}-\frac{\ell}{2}+1) \cdots (\frac{n}{2}+\frac{\ell}{2}) \ell!}{(n+2) \cdots (n+\ell+1) (\frac{\ell}{2})!^2}\\
&= (\frac{1}{2})^{\ell} \frac{(n-\ell+2)(n-\ell+4)\cdots (n+\ell) \ell!}{(n+2) \cdots (n+\ell+1) \frac{\ell}{2}!^2}\\
&= (\frac{1}{2})^{\ell} \frac{\ell!}{ \frac{\ell}{2}!^2} (1-\frac{\ell}{n+2}) (1-\frac{\ell-1}{n+3})\cdots (1-\frac{1}{n+2+\ell}),
\end{align*}
which proves the first formula. The second formula follows from Corollaries~\ref{lem:SU2} and \ref{lem:qubitCG}. The estimate derives from
\begin{align*}
\frac{(n+1)!n!}{(n-\ell)!(n+1+\ell)!}= \frac{(n-\ell+1)\cdots n}{(n+2)\cdots (n+1+\ell)}&\geq \left(\frac{n-\ell+1}{n+2}\right)^{\ell}\geq 1-\frac{\ell(\ell+1)}{n+2}.
\end{align*}
\end{proof}

Finally, we compute the Fourier coefficients of the distribution that is uniform on the equator of the Bloch sphere. 
\begin{lemma}\label{lem:equator}
Let $d=2$ and let $\mu(x)$ be the distribution that is uniformly concentrated on the equator of the Bloch sphere, i.e. 
$$\int \mu(x) f(x) dx = \frac{1}{2\pi }\int_{[0, 2\pi)} f(x_{\phi}) d\phi$$
for all test functions $f(x)$, where 
$$x_\phi= \frac{(e^{i\phi/2}\ket{1}+e^{-i\phi/2}\ket{2})(e^{-i\phi/2}\bra{0}+e^{-i\phi/2}\bra{1})}{2}$$ are the points on the equator. Then
\begin{align*}
\mu(x)= \sum_{\ell} (-1)^{\frac{\ell}{2}}(\frac{1}{2})^{\ell} \binom{\ell}{\frac{\ell}{2}} y_{\ell,0}(x) 
\end{align*}
\end{lemma}
\begin{proof}
Note that $\mu(x)=\mu(hx)$ for all $h\in H=U(1)\times U(1)$, i.e. $\mu$ is $H$ invariant. The non-vanishing Fourier components must therefore also be $H$ invariant, which implies that only the ones where $m= 0$ can be nonzero.

The remaining coefficients are
\begin{align*}
\int \mu(x) \overline{y_{\lambda,0}(x)} dx &=\frac{1}{2\pi} \int_\phi  d\phi  \overline{y_{\lambda, 0}(x_\phi)}\\
 &=\frac{1}{2\pi} \int_\phi  d\phi \bra{\lambda, 0} g_\phi \ket{\lambda, 0}\\
&= \bra{\lambda, 0} \left(\begin{array}{cc}\frac{1}{\sqrt{2}} & i\frac{1}{\sqrt{2}}  \\ i\frac{1}{\sqrt{2}} & \frac{1}{\sqrt{2}} \end{array}\right) \ket{\lambda, 0}
\end{align*}
where we chose a representative $g_\phi$ from the coset $x_\phi H$: 
$$g_\phi=\left(\begin{array}{cc}e^{i \phi/2} & 0  \\  0 & e^{-i \phi/2} \end{array}\right) \left(\begin{array}{cc}\frac{1}{\sqrt{2}} & i\frac{1}{\sqrt{2}} \\  i\frac{1}{\sqrt{2}} & \frac{1}{\sqrt{2}} \end{array}\right)$$
and used $\bra{\lambda, 0}= \bra{\lambda, 0}\left(\begin{array}{cc}e^{i \phi/2} & 0  \\  0 & e^{-i \phi/2} \end{array}\right)$ in the last line. The calculation of the last term is a combinatorial feast:
\begin{align*}
\bra{\lambda, 0}& \left(\begin{array}{cc}\frac{1}{\sqrt{2}} &i \frac{1}{\sqrt{2}} \\ i \frac{1}{\sqrt{2}} & \frac{1}{\sqrt{2}} \end{array}\right) \ket{\lambda, 0}  =\bra{\lambda, 0} \left(\begin{array}{cc}\frac{1}{\sqrt{2}} & \frac{1}{\sqrt{2}} \\ \frac{1}{\sqrt{2}} & -\frac{1}{\sqrt{2}} \end{array}\right) \ket{\lambda, 0} \\
& =\bra{\lambda, 0} \begin{pmatrix} e^{-i\pi/4} & 0 \\ 0 & e^{i\pi/4} \end{pmatrix}\left(\begin{array}{cc}\frac{1}{\sqrt{2}} & \frac{1}{\sqrt{2}} \\ \frac{1}{\sqrt{2}} & -\frac{1}{\sqrt{2}} \end{array}\right)  \begin{pmatrix} e^{-i\pi/4} & 0 \\ 0 & e^{i\pi/4} \end{pmatrix}\ket{\lambda, 0} \\
& =\frac{1}{\binom{2\ell}{\ell}}\sum_{z, z'} \bra{z} \left(\begin{array}{cc}\frac{1}{\sqrt{2}} & \frac{1}{\sqrt{2}}  \\\frac{1}{\sqrt{2}} & -\frac{1}{\sqrt{2}} \end{array}\right)^{\otimes 2\ell} \ket{z'}\\
&=\sum_{z} \bra{z} \left(\begin{array}{cc}\frac{1}{\sqrt{2}} & \frac{1}{\sqrt{2}}  \\\frac{1}{\sqrt{2}} & -\frac{1}{\sqrt{2}} \end{array}\right)^{\otimes 2\ell} \ket{1 \cdots 1 2 \cdots 2}\\
&=\left(\frac{1}{\sqrt{2}}\right)^{2\ell} \sum_{z} \bra{z} \ket{1+2}^{\otimes \ell}\ket{1-2}^{\otimes \ell}\\
&=\frac{1}{2^{\ell}} \frac{1}{\ell!^2}\left( \frac{d}{d\alpha}\right)^{\ell} \left( \frac{d}{d\beta}\right)^{\ell}(\alpha+\beta)^{\ell} (\alpha-\beta)^{\ell}\big|_{\alpha=\beta=0} \\
&=\frac{1}{2^{\ell}} \frac{1}{\ell!^2}\left( \frac{d}{d\alpha}\right)^{\ell} \left( \frac{d}{d\beta}\right)^{\ell}(\alpha^2-\beta^2)^{\ell}\big|_{\alpha=\beta=0} \\
&=\frac{1}{2^{\ell} } \sum_{z} \bra{z} \ket{11-22}^{\otimes \ell}\\
&=\frac{1}{2^{\ell}} (-1)^{\frac{\ell}{2}}\binom{\ell}{\frac{\ell}{2}}
\end{align*}
if $\ell$ is even. Otherwise the last formula vanishes. The sums over $z$ and $z'$ extend over all binary strings (with symbols $1$ and $2$) of length $ 2\ell$ with Hamming weight $\ell$. 
\end{proof}

\section{D $\quad$ Examples}

\subsection{Holevo's Covariant Measurement}

It was shown by Holevo that an optimal measurement procedure (in terms of the fidelity) for state estimation is given by the POVM $\{\proj{y}^{\otimes n}dy\}$~\cite[p. 163]{Holevo}. We now want to analyse this measurement with our methods. Let us start by assuming that we have measured the effect $\proj{d}^{\otimes n}$, giving rise to an estimate density
\begin{align*}
\mu_{\proj{d}^{\otimes n}}(x)& = \dim(n, d)|\braket{x}{d}|^{2n}
\end{align*}
We find
\begin{align}\label{eq:convergence}
\lim_{n\rightarrow \infty}\mu_{\rho^k_{\proj{d}^{\otimes n}}}(x)=\delta(x),
\end{align}
since $\dim(n, d)|\braket{x}{d}|^{2n}$ converges to the $\delta$-distribution. This, as expected, reflects the fact that the scheme is asymptotically correct. If we measured $\proj{z}^{\otimes n}$ instead of $\proj{d}^{\otimes n}$, we find 
\begin{align}\label{eq:optimalmeasurement}
\mu_{\proj{z}^{\otimes n}}(x)& =\dim(n, d) |\braket{x}{z}|^{2n}
\end{align}
with
$$\lim_{n\rightarrow \infty} \mu_{\rho^k_{\proj{z}^{\otimes n}}}(x)= \delta(z-x).$$

Let us now do the analysis in terms of Fourier coefficients: Comparing~\eqref{eq:convergence} with  $\delta(x)= \sum_{
\ell} y_{\ell, 0}(x)$
we see that the Fourier coefficients of $\dim(n, d)|\braket{d}{x}|^{2n}$ must all converge to one. Explicitly, the latter are given by (see Corollary~\ref{cor:productexpansion2})
\begin{align}
\dim(n, d)|\braket{d}{x}|^{2n}=\frac{1}{\dim(n, d)}\sum_{\ell}^n  \CG{\nu}{\nu^*}{\la}{0}{0}{0}  y_{\ell, 0}(x),
\end{align}
where $\la=(\ell, 0, \ldots, 0, -\ell)$, $\nu=(n, 0, \ldots, 0)$ and $\nu^*$ denotes highest weight dual to $\nu$. More generally, we have (see Corollary~\ref{cor:productexpansion3})
\begin{align}
\dim(n, d) |\braket{z}{x}|^{2n}& =\frac{1}{\dim(n, d)} \sum_{\ell}^n
\CG{\nu}{\nu^*}{\la}{0}{0}{0} \sum_m \overline{y_{\ell, m}(z)}y_{\ell,
  m}(x)\ .
\end{align}

We conclude this example with a formula for the Fourier coefficients for the qubit case and derive from it explicit bounds on the convergence the value one (Corollary~\ref{cor:clebsch-gordon-estimate2}):
$$\frac{1}{\dim(n, 2)}  \CG{\nu}{\nu^*}{\la}{0}{0}{0} =\frac{n!(n+1)!}{(n-\ell)!(n+\ell+1)!},$$
for $\ell\leq n$ and zero otherwise. For small $n$ and $\ell$, we have 
\begin{center}
\begin{tabular}{cccccc}
    & $n$ & $1$ & $2$ & $3$ & $4$\\
 $l$  &    &    &    &    & \\
 $0$ &  & $1$ & $1$ & $1$ & $1$  \\
 $1$ & & $\frac{1}{3}$ 	& $\frac{5}{10}$ & $\frac{21}{35}$ & $\frac{84}{126}$  \\
 $2$ & & 			& $\frac{1}{10}$ & $\frac{7}{35}$ & $\frac{36}{126}$  \\
 $3$ & &			&                    & $\frac{1}{35}$ & $\frac{9}{126}$  \\
 $4$ & & 			&                    & 		   & $\frac{1}{126}$  \\
\end{tabular}
\end{center}
and in general there is the bound (Corollary~\ref{cor:clebsch-gordon-estimate2})
$$1-\frac{\ell(\ell+1)}{n+2}\leq\frac{1}{\dim(n, 2)}  \CG{\nu}{\nu^*}{\la}{0}{0}{0} \leq 1.$$
This example shows how one may perform a convergence analysis of a tomographic measurement in terms of the Fourier coefficients of the of the estimate density. The convergence of the Fourier coefficients to a constant value (for fixed $\ell$ and $n\rightarrow \infty$) is also in agreement with our intuition about the duality of the Fourier transform: more information about $x$ corresponds to less information about $(\ell, m)$.

\subsection{Basis Measurements}
We will now analyse the case where a product measurement is carried out with the measurement in a single system given by an orthonormal basis.

Assume that the basis in which we measure is the computational basis $\{\proj{i}\}_{i}^d$. Since the state we are measuring lives in the symmetric subspace, we can, without loss of generality, project the effect $B_f^n$ onto this subspace and obtain the projector onto the vector with weight $f$ in the representation $V_\nu$. This vector, denoted by $\ket{\nu, m}$ is unique and has Gelfand-Zetlin pattern $m=(m^{(d-1)}, \ldots, m^{(1)})$ for $m^{(i)}=(\sum_{j=d+1-i}^d f^{(j)}, 0, \ldots, 0)$. The estimate density is therefore given by
$$\mu_{B^n_f}(x)=\dim(n, d)|\braket{x}{\nu, m}|^2,$$
where $\nu=(n, 0, \ldots, 0)$.
By Lemma~\ref{lem:productexpansion} we have
$$\mu_{B^n_f}(\ell, m')=\delta_{m', 0} \frac{1}{\dim(n, d)}\CG{\nu}{\nu^*}{\la}{m}{-m}{0},$$
where $\la=(\ell, 0, \ldots, 0, -\ell)$.

We will now compute these coefficients for qubits, $d=2$, in the case where we have measured an equal number, namely, $\frac{n}{2}$, 1s and 2s (i.e. the case $m=\frac{n}{2}$). We will use this formula to show that estimate density converges to the uniform distribution on the equator of the Bloch sphere. It follows from Corollary~\ref{cor:clebsch-gordon-estimate2} that for $\ell$ and $n/2$ even:
\begin{align*} 
\mu_{B^n_f}(\ell, m)=\delta_{m, 0}(-1)^{\frac{\ell}{2}} (\frac{1}{2})^{\ell} \binom{\ell}{\frac{\ell}{2}} \prod_{i=1}^{\ell}\left( 1-\frac{\ell-i}{n+2+i}\right).
\end{align*}
For large $n$, these coefficients turn into 
$$(-1)^{\frac{\ell}{2}} (\frac{1}{2})^{\ell} \binom{\ell}{\frac{\ell}{2}}$$
which are the Fourier coefficients of the uniform distribution on the equator of the Bloch sphere by Lemma~\ref{lem:equator}. The estimate density therefore concentrates on the equator just as expected, since we do not obtain any information on the phase of the state from this measurement.

This example shows that Fourier analysis is able to trace a complicated convergence behaviour in a compact way. When several bases are used (such as in the BB84 or the six-state protocols for quantum key distribution) one can use the just derived formula together Lemma~\ref{lem:multiply} --- which allows to rotate the basis --- in the update rule.

\end{document}